\documentclass[11pt]{article}
\usepackage{amssymb}
\usepackage{amsfonts}
\usepackage{indentfirst}
\usepackage{amsopn}
\usepackage{amsmath}
\usepackage{framed}
\usepackage{amsthm}
\usepackage{amscd}
\usepackage{graphicx}
\usepackage[dvips]{epsfig}
\usepackage{caption}
\usepackage{subcaption}
\usepackage{epstopdf}

\makeatletter\@addtoreset {equation}{section}\makeatother

\newtheorem{theorem}{Theorem}

\newtheorem{lemma}{Lemma}
\newtheorem{remark}{Remark}
\newtheorem{definition}{Definition}
\newtheorem{example}{Example}

\begin{document}

\title{\bf Periodic travelling waves of the modified KdV equation and
rogue waves on the periodic background}

\author{Jinbing Chen$^1$ and Dmitry E. Pelinovsky$^{2}$ \\
{\small \it $^1$ School of Mathematics, Southeast University, Nanjing, Jiangsu 210096, P.R. China} \\
{\small \it $^2$ Department of Mathematics, McMaster University, Hamilton, Ontario, Canada, L8S 4K1 } }

\date{\today}
\maketitle

\begin{abstract}
We address the most general periodic travelling wave of the modified Korteweg--de Vries (mKdV) equation
written as a rational function of Jacobian elliptic functions. By applying an algebraic method
which relates the periodic travelling waves and the squared periodic
eigenfunctions of the Lax operators, we characterize {\em explicitly} the location of eigenvalues
in the periodic spectral problem away from the imaginary axis.
We show that Darboux transformations with the periodic eigenfunctions remain in the class
of the same periodic travelling waves of the mKdV equation.
In a general setting, there exist three symmetric pairs of simple eigenvalues
away from the imaginary axis, and we give a new representation of the second non-periodic
solution to the Lax equations for the same eigenvalues. We show that Darboux transformations
with the non-periodic solutions to the Lax equations produce rogue waves on the periodic background, which are
either brought from infinity by propagating algebraic solitons or formed in a finite region of the time-space plane.
\end{abstract}


\section{Introduction}

We address periodic travelling waves of the modified Korteweg--de Vries (mKdV) equation,
which we take in the normalized form:
\begin{equation}\label{mKdV}
u_t+6u^2u_x+u_{xxx}=0.
\end{equation}
As is well-known since the pioneer paper \cite{AKNS},
the mKdV equation (\ref{mKdV}) is a compatibility condition of the following
pair of two linear equations written for the vector $\varphi = (\varphi_1,\varphi_2)^t$:
\begin{equation}\label{3.2}
\varphi_x = U(\lambda,u) \varphi, \quad
U(\lambda,u) =\left(\begin{array}{cc} \lambda&u\\ -u&-\lambda\\ \end{array}\right),
\end{equation}
and
\begin{equation}\label{3.3}
\varphi_t = V(\lambda,u) \varphi, \quad
V(\lambda,u) = \left(\begin{array}{cc}
-4\lambda^3-2\lambda u^2&-4\lambda^2u-2\lambda u_x-2u^3-u_{xx}\\
4\lambda^2u-2\lambda u_x+2u^3+u_{xx}&4\lambda^3+2\lambda u^2\\
\end{array}\right).
\end{equation}
Assuming $\varphi(x,t) \in C^{2,2}(\mathbb{R} \times \mathbb{R})$ and $u(x,t) \in C^{3,1}(\mathbb{R} \times \mathbb{R})$,
the compatibility condition $\varphi_{xt} = \varphi_{tx}$ is equivalent to the mKdV equation
(\ref{mKdV}) satisfied in the classical sense.

Among the periodic travelling wave solutions, the mKdV equation (\ref{mKdV})
admits the normalized constant wave $u(x,t) = 1$ and two families of the normalized periodic waves given by
\begin{equation}
\label{dn-intro}
u(x,t) = {\rm dn}(x-ct;k), \quad c = 2-k^2
\end{equation}
and
\begin{equation}
\label{cn-intro}
u(x,t) = k {\rm cn}(x-ct;k), \quad c = 2k^2 - 1,
\end{equation}
where ${\rm dn}$ and ${\rm cn}$ are Jacobian elliptic functions
and $k \in (0,1)$ is the elliptic modulus (see Chapter 8.1 in \cite{Grad}
for review of elliptic functions and integrals).

The normalized constant wave $u(x,t) = 1$ is linearly and nonlinearly stable
in the time evolution of the mKdV equation (\ref{mKdV}) in the sense that any small
perturbation to the constant wave in the energy space $H^1(\mathbb{R})$ remains
small in the $H^1(\mathbb{R})$ norm globally in time, see, e.g., \cite{KPV,KT}.
Among the exact solutions to the mKdV equation (\ref{mKdV}) on the normalized constant wave,
we note the following {\em algebraic soliton}
\begin{equation}
\label{alg-soliton}
u(x,t) = 1 - \frac{4}{1 + 4 (x-6t-x_0)^2},
\end{equation}
where $x_0 \in \mathbb{R}$ is arbitrary. The algebraic
soliton propagates on the normalized constant background with the speed $c_0 = 6$.

In our previous work \cite{CPkdv}, we have constructed new solutions on the periodic background
given by the normalized periodic waves (\ref{dn-intro}) and (\ref{cn-intro}). In doing so,
we have adopted the formal algebraic method from \cite{Cao1,Cao2,Cao3} and elaborated
the following algorithm for constructing new solutions to the mKdV equation (\ref{mKdV}):

\begin{samepage}
\begin{framed}
\begin{enumerate}
\item Impose a constraint between a solution $u$ to the mKdV equation (\ref{mKdV})
and a solution $\varphi = (p_1,q_1)^t$ to the Lax system (\ref{3.2})--(\ref{3.3}) with $\lambda = \lambda_1$ and deduce
closed differential equations on $u$. These equations are satisfied if $u$ is a periodic travelling wave.
\item Characterize the set of admissible values for $\lambda_1$ and the relations between $u$ and
the squared components $p_1^2 + q_1^2$, $p_1^2 - q_1^2$, and $p_1 q_1$, The solution $\varphi = (p_1,q_1)^t$
is periodic in $x$ and is travelling in $t$.
\item Obtain the second solution $\varphi = (\hat{p}_1,\hat{q}_1)^t$ to the Lax system (\ref{3.2})--(\ref{3.3})
for the same values of $\lambda_1$. The second solution is non-periodic, it grows linearly in $x$ and $t$
almost everywhere as $|x| + |t| \to \infty$.
\item Apply Darboux transformation with the second solution $\varphi = (\hat{p}_1,\hat{q}_1)^t$ and obtain new solutions
to the mKdV equation (\ref{mKdV}) on the periodic background $u$.
\end{enumerate}
\end{framed}
\end{samepage}

As the main outcome of step 2 in \cite{CPkdv}, we obtained two pairs of admissible values
for $\lambda_1$ with ${\rm Re}(\lambda_1) \neq 0$. For the {\rm dn}-periodic wave (\ref{dn-intro}), the two pairs are real
$\pm \lambda_+$ and $\pm \lambda_-$ with
\begin{equation}
\label{eig-dn-intro}
\lambda_{\pm} = \frac{1}{2} (1 \pm \sqrt{1-k^2}).
\end{equation}
For the {\rm cn}-periodic wave (\ref{cn-intro}), the two pairs are complex-conjugate
$\pm \lambda_+$ and $\pm \lambda_-$ with
\begin{equation}
\label{eig-cn-intro}
\lambda_{\pm} = \frac{1}{2} (k \pm i \sqrt{1-k^2}).
\end{equation}

As the main outcome of step 4 in \cite{CPkdv}, we constructed an algebraic soliton propagating on the background
of the {\rm dn}-periodic wave (\ref{dn-intro}) and a fully localized rogue wave on the $(x,t)$ plane growing and decaying
on the background of the {\rm cn}-periodic wave (\ref{cn-intro}). Since the ${\rm dn}$-periodic wave (\ref{dn-intro})
converges to the constant wave $u(x,t) = 1$ as $k \to 0$, the algebraic soliton on the ${\rm dn}$-periodic wave background
generalizes the exact solution (\ref{alg-soliton}). The rogue wave on the ${\rm cn}$-periodic wave
background satisfies the following mathematical definition of a rogue wave.

\begin{definition}
Let $u$ be a periodic travelling wave of the mKdV equation (\ref{mKdV}) with the period $L$
and $\tilde{u}$ be another solution to the mKdV equation (\ref{mKdV}). We say that
$\tilde{u}$ is a rogue wave on the background $u$ if $\tilde{u}$ is different from
the orbit $\{ u(x-x_0) \}_{x_0 \in [0,L]}$ for $t \in \mathbb{R}$ and it satisfies
\begin{equation}
\label{rogue-wave-def}
\inf_{x_0 \in [0,L]} \sup_{x \in \mathbb{R}} \left| \tilde{u}(x,t) - u(x-x_0) \right| \to 0 \quad \mbox{\rm as} \quad t \to \pm \infty.
\end{equation}
\label{def-rogue}
\end{definition}

Definition \ref{def-rogue} correspond to the physical interpretation of a rogue wave as {\em the wave that
comes from nowhere and disappears without any trace}.
Rogue waves in physics are associated with the gigantic waves on the ocean's surface
and in optical fibers which arise due to the modulation instability of the background wave \cite{Charif,Wabnitz}.
Several recent publications were devoted to numerical and analytical studies of rogue waves
on the background of periodic \cite{AZ2,GS}, quasi-periodic \cite{Tovbis1,Tovbis2,CalSch},
and multi-soliton \cite{BM1,BM2,BM3} wave patterns in the framework of the nonlinear Schr\"{o}dinger
(NLS) equation. Since internal waves are modeled by
the mKdV equation \cite{Grimshaw}, formation of rogue internal waves was also studied in the framework of
the mKdV equation (\ref{mKdV}) as a result of multi-soliton interactions \cite{Sh1,PelSh,Sh2,PelSl}.

The difference between the two outcomes of the algorithm applied in \cite{CPkdv} to the normalized
periodic waves (\ref{dn-intro}) and (\ref{cn-intro})  in the mKdV equation is related to the fact that
the {\rm dn}-periodic waves are modulationally stable with respect to perturbations
of long periods, whereas the {\rm cn}-periodic waves are modulationally unstable \cite{BJK,BHJ}.

The purpose of this work is to consider the most general periodic travelling wave of the mKdV equation (\ref{mKdV})
and to characterize {\em explicitly} location of eigenvalues $\lambda$ with ${\rm Re}(\lambda) \neq 0$ in the periodic spectral problem (\ref{3.2}).
Although it may seem to be an incremental goal, advancement from the normalized periodic waves (\ref{dn-intro}) and (\ref{cn-intro})
to the most general periodic wave of the mKdV equation (\ref{mKdV}) require us to consider Riemann Theta functions
of genus $g = 2$, which are expressed as rational functions of Jacobian elliptic functions.
As is well known \cite{MatveevBook,Gestezy}, Riemann Theta functions of genus $g$ represent quasi-periodic
solutions to many integrable evolution equations including the mKdV equation (\ref{mKdV}). Hence, having successfully solved
the problem for $0 \leq g \leq 2$, we can move to the next goal of solving this problem for general $g$.

The algebraic method developed here is different from construction of multi-soliton solutions
on the background of quasi-periodic solutions developed in \cite{GSvirsky} by using commutation methods.
It is also different from other analytical techniques for explicit characterization of eigenvalues
related to the periodic solutions of the mKdV equation in the Whitham modulation theory \cite{K1,K2}
(see also \cite{K,P} for earlier works).

Let us now present the main results of our work. The travelling wave to the mKdV equation (\ref{mKdV})
has the form $u(x,t) = u(x-ct)$, where $c$ is wave speed. The wave profile $u$
satisfies the third-order differential equation:
\begin{eqnarray}
\label{third-order-intro}
\frac{d^3 u}{dx^3} + 6 u^2 \frac{du}{dx} - c \frac{du}{dx} = 0.
\end{eqnarray}
Integrating it once yields the second-order differential equation:
\begin{eqnarray}
\label{second-order-intro}
\frac{d^2 u}{dx^2} + 2 u^3 - c u = e,
\end{eqnarray}
where $e$ is the integration constant. Integrating it once again yields
the first-order invariant:
\begin{eqnarray}
\label{first-order-intro}
\left( \frac{du}{dx} \right)^2 + u^4 - c u^2 + d = 2 e u,
\end{eqnarray}
where $d$ is another integration constant. Thus, the most general periodic travelling
wave in the mKdV equation (\ref{mKdV}) is characterized by the parameters $(c,d,e)$.
The previous case considered in \cite{CPkdv} corresponds to $e = 0$.

Our first result is about classification of the most general periodic travelling
wave solution to the mKdV equation (\ref{mKdV}). As is well-known
(see, e.g., \cite{Vasiljev}), there exists two explicit families of the periodic solutions to
equations (\ref{second-order-intro}) and (\ref{first-order-intro})
depending on parameters $(c,d,e)$. When the polynomial
\begin{equation}
P(u) := u^4 - c u^2 + d - 2 e u,
\label{polynomial}
\end{equation}
admits four real roots ordered as $u_4 \leq u_3 \leq u_2 \leq u_1$, where $(u_1,u_2,u_3,u_4)$
are related to the parameters $(c,d,e)$,
the exact periodic solution to the system (\ref{second-order-intro}) and (\ref{first-order-intro}) is given by
\begin{equation}
u(x) = u_4 + \frac{(u_1-u_4) (u_2-u_4)}{(u_2-u_4) + (u_1 - u_2) {\rm sn}^2(\nu x;\kappa)},
\label{Jacob-1-intro}
\end{equation}
where $\nu > 0$ and $\kappa \in (0,1)$ are parameters given by
\begin{equation}
\label{roots-1-intro}
\left\{ \begin{array}{l}
4 \nu^2 = (u_1 - u_3) (u_2 - u_4), \\
4 \nu^2 \kappa^2 = (u_1 - u_2) (u_3 - u_4). \end{array} \right.
\end{equation}
When the polynomial $P(u)$ in (\ref{polynomial}) admits two real roots $b \leq a$
and two complex-conjugate roots $\alpha \pm i \beta$, where $(a,b,\alpha,\beta)$ are related to parameters
$(c,d,e)$, the exact periodic solution to the system (\ref{second-order-intro}) and (\ref{first-order-intro})
is given by
\begin{equation}
u(x) = a + \frac{(b-a) (1 - {\rm cn}(\nu x; \kappa))}{1 + \delta + (\delta - 1) {\rm cn}(\nu x;\kappa)},
\label{Jacob-2-intro}
\end{equation}
where $\delta > 0$, $\nu > 0$, and $\kappa \in (0,1)$ are parameters given by
\begin{equation}
\label{roots-2-intro}
\left\{ \begin{array}{l}
\delta^2 = \frac{(b-\alpha)^2 + \beta^2}{(a-\alpha)^2 + \beta^2}, \\
\nu^2 = \sqrt{\left[ (a-\alpha)^2 + \beta^2 \right] \left[ (b-\alpha)^2 + \beta^2 \right]}, \\
2 \kappa^2 = 1 - \frac{(a-\alpha) (b - \alpha) + \beta^2}{\sqrt{\left[ (a-\alpha)^2 + \beta^2 \right] \left[ (b-\alpha)^2 + \beta^2 \right]}}. \end{array} \right.
\end{equation}
The trivial case when the polynomial $P(u)$ in (\ref{polynomial})
admits no real roots does not produce any real solution to the system (\ref{second-order-intro}) and (\ref{first-order-intro}).
The following theorem characterizes the periodic travelling waves to the mKdV equation.

\begin{theorem}
\label{theorem-wave}
Fix $c > 0$ and $e \in (-e_0,e_0)$ with $e_0 := 2 \sqrt{c^3}/(3 \sqrt{6})$.
There exist $-\infty < d_1 < d_2 < \infty$ such that for every $d \in (d_1,d_2)$,
the system (\ref{second-order-intro}) and (\ref{first-order-intro})
admits the exact periodic solution in the form (\ref{Jacob-1-intro}) with (\ref{roots-1-intro})
and three other periodic solutions of the same period obtained with the
following three symmetry transformations
\begin{eqnarray}
\label{symm-intro}
\mbox{\rm (S1)} \quad u_1 \leftrightarrow u_2, \quad u_3 \leftrightarrow u_4, \quad
\mbox{\rm (S2)} \quad u_1 \leftrightarrow u_3, \quad u_2 \leftrightarrow u_4, \quad
\mbox{\rm (S3)} \quad u_1 \leftrightarrow u_4, \quad u_2 \leftrightarrow u_3.
\end{eqnarray}
In addition, if $e \neq 0$, there exists $d_3 > d_2$ such that for every $d \in (-\infty,d_1) \cup (d_2,d_3)$
the system  (\ref{second-order-intro}) and (\ref{first-order-intro})
admits the exact periodic solution in the form (\ref{Jacob-2-intro}) with (\ref{roots-2-intro})
and another periodic solution of the same period obtained with the symmetry transformation
\begin{eqnarray}
\label{symm-4-intro}
\mbox{\rm (S0)} \quad a \leftrightarrow b.
\end{eqnarray}
For every other value of $(c,e)$, there exists $d_1 > 0$ such that for every $d \in (-\infty,d_1)$
only the periodic solution in the form (\ref{Jacob-2-intro})--(\ref{roots-2-intro}) exists
together with another solution obtained by the symmetry transformation (\ref{symm-4-intro}).
All the solutions are unique up to the translational symmetry $u(x) \mapsto u(x+x_0)$, $x_0 \in \mathbb{R}$.
\end{theorem}

\begin{remark}
\label{remark-2}
Periodic solutions of the third-order equation (\ref{third-order-intro}) are invariant
with respect to the reflection $u \mapsto -u$. The reflection corresponds to
the transformation $e \mapsto -e$ in the second-order equation (\ref{second-order-intro}).
\end{remark}

\begin{remark}
The proof of Theorem \ref{theorem-wave} is elementary.
It is based on the phase-plane analysis and properties of the Jacobian elliptic functions.
We included Theorem \ref{theorem-wave} for clarity of our presentation.
\end{remark}

\begin{remark}
The periodic travelling wave of the mKdV equation (\ref{mKdV}) in Theorem \ref{theorem-wave}
is also the periodic travelling wave of the following Gardner equation:
\begin{equation}\label{Gardner}
v_t + 12 a v v_x + 6 v^2 v_x + v_{xxx}=0,
\end{equation}
where $a \in \mathbb{R}$ is arbitrary. Indeed, if $u(x,t) \in C^{3,1}(\mathbb{R} \times \mathbb{R})$ satisfies the mKdV
equation (\ref{mKdV}) and is represented by $u(x,t) = a + v(x-6a^2 t,t)$ with $a \in \mathbb{R}$,
then $v(x,t) \in C^{3,1}(\mathbb{R} \times \mathbb{R})$ satisfies the Gardner equation (\ref{Gardner}).
The Gardner equation is commonly used in modelling of internal waves \cite{Grimshaw}.
\end{remark}

Next, we use the algebraic method from \cite{Cao1,Cao2,Cao3} and relate the solutions $u$
in Theorem \ref{theorem-wave} with squared eigenfunctions of the periodic spectral problem (\ref{3.2}).
Compared to our previous work in \cite{CPkdv}, we have to use two squared eigenfunctions
for two different eigenvalues $\lambda$ in order to obtain periodic solutions
of the third-order differential equation (\ref{third-order-intro}).
The following theorem represents the outcome of the algebraic method.

\begin{theorem}
\label{theorem-eig}
The spectral problem (\ref{3.2}) with $u$ given by  the periodic waves (\ref{Jacob-1-intro})
and (\ref{Jacob-2-intro}) admits three pairs $\pm \lambda_1$, $\pm \lambda_2$, $\pm \lambda_3$ of eigenvalues
with ${\rm Re}(\lambda) \neq 0$ that corresponds to the periodic eigenfunctions $\varphi$ of the same period.
For the periodic wave (\ref{Jacob-1-intro}),
the eigenvalues are located at
\begin{eqnarray}
\label{eig-real-intro}
\lambda_1 = \frac{1}{2} (u_1 + u_2), \quad \lambda_2 = \frac{1}{2} (u_1+u_3), \quad
\lambda_3 = \frac{1}{2} (u_2 + u_3).
\end{eqnarray}
For the periodic wave (\ref{Jacob-2-intro}), the eigenvalues are located at
\begin{eqnarray}
\label{eig-complex-intro}
\lambda_1 = \frac{1}{4} (a-b) + \frac{i}{2} \beta, \quad \lambda_2 = \frac{1}{4} (a-b) - \frac{i}{2} \beta,
\quad \lambda_3 = \frac{1}{2} (a+b).
\end{eqnarray}
\end{theorem}

\begin{remark}
The algebraic method in the proof of Theorem \ref{theorem-eig} gives also
explicit characterization of the periodic eigenfunctions $\varphi$ of the spectral problem
(\ref{3.2}) at the three pairs of eigenvalues $\lambda$. If $u(x,t) = u(x-ct)$ is a travelling wave
solution to the mKdV equation (\ref{mKdV}), the time evolution
problem (\ref{3.3}) is also satisfied with the solution $\varphi(x,t) = \varphi(x-ct)$.
\end{remark}

\begin{remark}
The algebraic method does not allow us to conclude that no other eigenvalues $\lambda$
with ${\rm Re}(\lambda) \neq 0$ exist in the periodic spectral problem (\ref{3.2}).
\end{remark}

Next, we proceed with the multi-fold Darboux transformations \cite{GuHuZhou,Matveev} by using
the general form proven with explicit computations in Appendix A of \cite{CPkdv}.
Since we only have up to three pairs of simple eigenvalues in Theorem \ref{theorem-eig} in a generic
case, we should only use one-fold, two-fold, and three-fold Darboux transformations.
The corresponding Darboux transformations are given by the explicit expressions:
\begin{equation}
\label{one-fold-general}
\tilde{u} = u + \frac{4\lambda_1 p_1 q_1}{p_1^2 + q_1^2},
\end{equation}
\begin{equation}
\label{two-fold-general}
\tilde{u} = u + \frac{4 (\lambda_1^2-\lambda_2^2) \left[ \lambda_1 p_1 q_1 (p_2^2 + q_2^2) - \lambda_2 p_2 q_2 (p_1^2 + q_1^2)\right]}{
(\lambda_1^2 + \lambda_2^2) (p_1^2 + q_1^2)(p_2^2 + q_2^2) - 2 \lambda_1 \lambda_2 [ 4 p_1 q_1 p_2 q_2 + (p_1^2 - q_1^2)(p_2^2 - q_2^2)]},
\end{equation}
and
\begin{equation}
\label{three-fold-general}
\tilde{u} = u + \frac{4 N}{D},
\end{equation}
with
\begin{eqnarray*}
N & := & (\lambda_3^2 - \lambda_1^2) (\lambda_3^2 - \lambda_2^2) \lambda_3 p_3 q_3 \left[ (\lambda_1^2 + \lambda_2^2) (p_1^2 + q_1^2) (p_2^2 + q_2^2)
- 2 \lambda_1 \lambda_2 (p_1^2 - q_1^2) (p_2^2 - q_2^2) \right] \\
& \phantom{t} & + (\lambda_2^2 - \lambda_1^2) (\lambda_2^2 - \lambda_3^2) \lambda_2 p_2 q_2
\left[ (\lambda_1^2 + \lambda_3^2) (p_1^2 + q_1^2) (p_3^2 + q_3^2)
- 2 \lambda_1 \lambda_3 (p_1^2 - q_1^2) (p_3^2 - q_3^2) \right] \\
& \phantom{t} & + (\lambda_1^2 - \lambda_2^2) (\lambda_1^2 - \lambda_3^2) \lambda_1 p_1 q_1
\left[ (\lambda_2^2 + \lambda_3^2) (p_2^2 + q_2^2) (p_3^2 + q_3^2)
- 2 \lambda_2 \lambda_3 (p_2^2 - q_2^2) (p_3^2 - q_3^2) \right] \\
& \phantom{t} & -8 (\lambda_1^4 + \lambda_2^4 + \lambda_3^4 - \lambda_1^2 \lambda_2^2
- \lambda_1^2 \lambda_3^2 - \lambda_2^2 \lambda_3^2 )\lambda_1 \lambda_2 \lambda_3 p_1 p_2 p_3 q_1 q_2 q_3
\end{eqnarray*}
and
{\small \begin{eqnarray*}
&& D := (\lambda_1 + \lambda_2)^2 (\lambda_1 + \lambda_3)^2 (\lambda_2 - \lambda_3)^2 (p_1^2 q_2^2 q_3^2 + q_1^2 p_2^2 p_3^2)
+ (\lambda_1 + \lambda_2)^2 (\lambda_2 + \lambda_3)^2 (\lambda_1 - \lambda_3)^2 (p_2^2 q_1^2 q_3^2 + q_2^2 p_1^2 p_3^2) \\
\phantom{t} && \phantom{t} + (\lambda_1 + \lambda_3)^2 (\lambda_2 + \lambda_3)^2 (\lambda_1 - \lambda_2)^2 (p_3^2 q_1^2 q_2^2 + q_3^2 p_1^2 p_2^2)
+ (\lambda_1 - \lambda_2)^2 (\lambda_1 - \lambda_3)^2 (\lambda_2 - \lambda_3)^2 (p_1^2 p_2^2 p_3^2 + q_1^2 q_2^2 q_3^2)  \\
\phantom{t} && \phantom{t} - 8 (\lambda_3^2 - \lambda_1^2) (\lambda_3^2 - \lambda_2^2) \lambda_1 \lambda_2 p_1 p_2 q_1 q_2 (p_3^2 + q_3^2)
- 8 (\lambda_2^2 - \lambda_1^2) (\lambda_2^2 - \lambda_3^2) \lambda_1 \lambda_3 p_1 p_3 q_1 q_3 (p_2^2 + q_2^2) \\
\phantom{t} && \phantom{t} - 8 (\lambda_1^2 - \lambda_2^2) (\lambda_1^2 - \lambda_3^2) \lambda_2 \lambda_3 p_2 p_3 q_2 q_3 (p_1^2 + q_1^2),
\end{eqnarray*}}where $\tilde{u}$ stands for a new solution to the mKdV equation (\ref{mKdV})
and $\varphi_j = (p_j,q_j)^t$, $j = 1,2,3$ stands for a nonzero solution to the Lax system
(\ref{3.2})--(\ref{3.3}) with potential $u$ and eigenvalue $\lambda_j$ assuming $\lambda_i \neq \pm \lambda_j$ for $i \neq j$.

The following theorems represent the outcomes of the Darboux transformations (\ref{one-fold-general}),
(\ref{two-fold-general}), and (\ref{three-fold-general}) with the periodic solutions to the Lax system (\ref{3.2})--(\ref{3.3}).

\begin{theorem}
\label{theorem-DT-dn}
Assume $u_4 < u_3 < u_2 < u_1$ such that $u_1 + u_2 + u_3 + u_4 = 0$,
$u_1 + u_2 \neq 0$, $u_1 + u_3 \neq 0$, and $u_2 + u_3 \neq 0$.
The one-fold Darboux transformation with the periodic eigenfunctions for each eigenvalue in (\ref{eig-real-intro})
transforms the periodic wave  (\ref{Jacob-1-intro}) to the periodic wave of the same period
obtained after the corresponding symmetry transformation in (\ref{symm-intro})
and the reflection $u \mapsto -u$.
The two-fold Darboux transformation with the periodic eigenfunctions
for any two eigenvalues from (\ref{eig-real-intro})
transforms the periodic wave (\ref{Jacob-1-intro}) to the periodic wave of the same period
obtained after the complementary third symmetry transformation in (\ref{symm-intro}).
The three-fold Darboux transformation with the periodic eigenfunctions
for all three eigenvalues in (\ref{eig-real-intro})
maps the periodic wave (\ref{Jacob-1-intro}) to itself reflected with $u \mapsto -u$.
\end{theorem}

\begin{remark}
Under the term ``complementary third transformation" in Theorem \ref{theorem-DT-dn}, we mean that if
$(\lambda_1,\lambda_2)$ is selected in (\ref{eig-real-intro}),
then (S3) is selected in (\ref{symm-intro}), and so on.
\end{remark}

\begin{remark}
The three eigenvalues in (\ref{eig-real-intro}) with the three periodic eigenfunctions
used in the one-fold transformation (\ref{one-fold-general}) recover
the three symmetry transformations in (\ref{symm-intro}). It is interesting to note that
due to the constraint $u_1 + u_2 + u_3 + u_4 = 0$, the choice of each eigenvalue indicate
explicitly the choice of the transformation between $u_1$, $u_2$, $u_3$, and $u_4$,
e.g., $\lambda_1 = (u_1+u_2)/2$ corresponds to (S1) with
$u_1 \leftrightarrow u_2$ and $u_3 \leftrightarrow u_4$, and so on.
\end{remark}

\begin{theorem}
\label{theorem-DT-cn}
Assume $a \neq \pm b$.
The one-fold Darboux transformation with the periodic eigenfunction
for eigenvalue $\lambda_3$ in (\ref{eig-complex-intro})
transforms the periodic wave  (\ref{Jacob-2-intro}) to the periodic wave of the same period
obtained after the symmetry transformation in (\ref{symm-4-intro})
and the reflection $u \mapsto -u$.
The two-fold Darboux transformation with the periodic eigenfunctions
for two eigenvalues $\lambda_1$ and $\lambda_2$ in (\ref{eig-complex-intro})
transforms the periodic wave (\ref{Jacob-2-intro}) to the periodic wave of the same period
obtained after the symmetry transformation in (\ref{symm-4-intro}).
The three-fold Darboux transformation with the periodic eigenfunctions
for all three eigenvalues in (\ref{eig-complex-intro})
maps the periodic wave (\ref{Jacob-2-intro}) to itself reflected with $u \mapsto -u$.
\end{theorem}

\begin{remark}
Other possible one-fold and two-fold Darboux transformations missing in the formulation of Theorem \ref{theorem-DT-cn}
returns complex-valued solutions $u$ to the mKdV equation (\ref{mKdV}).
\end{remark}

Finally, we use the algorithm above and construct new solutions
to the mKdV equation (\ref{mKdV}). To do so, we obtain the closed form expression for
the second solution $\varphi = (\hat{p}_1,\hat{q}_1)^t$
of the Lax system (\ref{3.2})--(\ref{3.3}) with $\lambda = \lambda_1$ given by an eigenvalue
in Theorem \ref{theorem-eig}. The explicit representation
of the second solution has been already obtained in our previous work \cite{CPkdv}, however,
the expression obtained there is singular at any point of $(x,t)$ where one
of the two components of the periodic solution $\varphi = (p_1,q_1)^t$ vanishes. As a result, the corresponding
expression was only used for the eigenvalue $\lambda_+$ in (\ref{eig-dn-intro}) but not
for the eigenvalue $\lambda_-$. Here we obtain a different explicit
representation of the second solution $\varphi$ which is free of singularities for every
eigenvalue $\lambda$ in Theorem \ref{theorem-eig}.
As a result, we are able to construct new solutions to the mKdV equation
by using the Darboux transformations for every second solution to the Lax system
(\ref{3.2})--(\ref{3.3}) corresponding to eigenvalues in Theorem \ref{theorem-eig}.
The following two theorems describe the new solutions as a function of $(x,t)$.

\begin{theorem}
\label{theorem-rogue-dn}
Assume $u_4 < u_3 < u_2 < u_1$ such that $u_1 + u_2 + u_3 + u_4 = 0$,
$u_1 + u_2 \neq 0$, $u_1 + u_3 \neq 0$, and $u_2 + u_3 \neq 0$.
Under three non-degeneracy conditions (\ref{nondeg1}), (\ref{nondeg2}), and (\ref{nondeg3}) below,
there exist $c_1,c_2,c_3 \neq c$ such that the second solutions to the Lax system (\ref{3.2})--(\ref{3.3})
for the eigenvalues $\lambda_1,\lambda_2,\lambda_3$ in (\ref{eig-real-intro}) are linearly growing in $x$ and $t$ everywhere
on the $(x,t)$ plane except for the straight lines $x - c_{1,2,3} t = \xi_{1,2,3}$, where
$\xi_1,\xi_2,\xi_3$ are phase parameters which are not uniquely defined.
The one-fold Darboux transformation with the second solution
for each eigenvalue $\lambda_{1,2,3}$ in (\ref{eig-real-intro}) adds an algebraic soliton
with the corresponding speed $c_{1,2,3}$ on the background of the periodic wave (\ref{Jacob-1-intro}) transformed by
the symmetry $(S_{1,2,3})$ and reflection $u \mapsto -u$.
The two-fold Darboux transformation with the
second solutions for any two eigenvalues from (\ref{eig-real-intro})
adds two algebraic solitons with the corresponding two wave speeds
on the background of the periodic wave (\ref{Jacob-1-intro}) transformed by
the complementary symmetry.
The three-fold Darboux transformation with the second solutions
for all three eigenvalues in (\ref{eig-real-intro}) adds
all three algebraic solitons with the three wave speeds
on the background of the periodic wave (\ref{Jacob-1-intro}) reflected with $u \mapsto -u$.
\end{theorem}

\begin{theorem}
\label{theorem-rogue-cn}
Assume $a \neq \pm b$. Under two non-degeneracy conditions (\ref{nondeg4}) and (\ref{nondeg5}) below,
the second solutions to the Lax system (\ref{3.2})--(\ref{3.3}) for
the eigenvalues $\lambda_1$ and $\lambda_2$ in (\ref{eig-complex-intro}) are linearly
growing in $x$ and $t$ everywhere, whereas there exists $c_0 \neq c$ such that the second solution
to the Lax system (\ref{3.2})--(\ref{3.3}) for the eigenvalue $\lambda_3$ in (\ref{eig-complex-intro}) is linearly growing in $x$ and $t$ everywhere
except for the straight line $x - c_0 t = \xi_0$, where $\xi_0$ is the phase parameter
which is not uniquely defined. The one-fold Darboux transformation with the second solution
for eigenvalue $\lambda_3$ in (\ref{eig-complex-intro}) adds an algebraic soliton
with the wave speed $c_0$ on the background of the periodic wave (\ref{Jacob-2-intro}) transformed by
the symmetry $(S_0)$ and reflection $u \mapsto -u$.
The two-fold Darboux transformation with the second solutions for two eigenvalues
$\lambda_1$ and $\lambda_2$ in (\ref{eig-complex-intro}) adds a rogue wave
on the background of the periodic wave (\ref{Jacob-2-intro}) transformed by
the symmetry $(S_0)$. The three-fold Darboux transformation with the second solutions
for all three eigenvalues in (\ref{eig-complex-intro}) adds
both the algebraic soliton with the wave speed $c_0$ and the
rogue wave on the background of the periodic wave (\ref{Jacob-2-intro}) reflected with $u \mapsto -u$.
\end{theorem}

\begin{remark}
The only rogue wave in Theorem \ref{theorem-rogue-cn} satisfies Definition \ref{def-rogue}, whereas
the other two solutions do not satisfy the limit (\ref{rogue-wave-def}) due to
the algebraic solitons propagating on the periodic
background along the straight line $x - c_0 t = \xi_0$. None of new solutions
in Theorem \ref{theorem-rogue-dn} satisfy Definition \ref{def-rogue} due to the same reason.
\end{remark}

The paper is organized as follows. Section \ref{sec-2} describes the algebraic method which allows us
to complete step 1 in the algorithm above. Section \ref{sec-3} characterizes the most general
periodic wave of the mKdV equation (\ref{mKdV}) and gives the proof of Theorem \ref{theorem-wave}.
Step 2 of the algorithm and the proof of Theorem \ref{theorem-eig} are given in Section \ref{sec-4}.
The proof of Theorems \ref{theorem-DT-dn} and \ref{theorem-DT-cn} can be found in Section \ref{sec-5}.
Step 3 of the algorithm is developed in Section \ref{sec-6}, where the second non-periodic solution
to the Lax system (\ref{3.2})--(\ref{3.3}) is obtained.
Step 4 of the algorithm is completed in Section \ref{sec-7}, where the new solutions to the mKdV equation (\ref{mKdV}) are constructed,
Theorems \ref{theorem-rogue-dn} and \ref{theorem-rogue-cn} are proven, and the surface plots
are included for graphical illustrations. Finally, Section \ref{sec-8} contains
a conjecture on the possible generalization of our results to
the case of quasi-periodic solutions to the mKdV equation (\ref{mKdV}).

\section{The algebraic method}
\label{sec-2}

We have showed in \cite{CPkdv} that imposing a constraint $u = p_1^2 + q_1^2$
between a solution $u$ to the mKdV equation (\ref{mKdV}) and
components of the solution $\varphi = (p_1,q_1)^t$ to the Lax system (\ref{3.2})--(\ref{3.3}) with
fixed parameter $\lambda = \lambda_1$ results in the second-order
differential equation on $u$. This equation is given by (\ref{second-order-intro}) for $e = 0$ and
it does not recover the most general periodic wave of the mKdV equation (\ref{mKdV}).
Here we extend the algebraic method of \cite{Cao1,Cao2,Cao3} by
imposing the constraint
\begin{equation}
\label{potential}
u = p_1^2 + q_1^2 + p_2^2 + q_2^2
\end{equation}
between a solution $u$ to the mKdV equation (\ref{mKdV}) and
components of two solutions $\varphi = (p_1,q_1)^t$ and $\varphi = (p_2,q_2)^t$ to
the Lax system (\ref{3.2})--(\ref{3.3}) with two
fixed parameters $\lambda = \lambda_1$ and $\lambda = \lambda_2$.
Assume $\lambda_1 \neq \pm \lambda_2$.

The spectral problem (\ref{3.2}) for two vectors $\varphi = (p_1,q_1)^t$
and $\varphi = (p_2,q_2)^t$ at $\lambda_1$ and $\lambda_2$ can be written as the Hamiltonian system
of degree two
\begin{eqnarray}
\label{Ham-sys}
\frac{d p_j}{dx} = \frac{\partial H_0}{\partial q_j}, \quad
\frac{d q_j}{dx} = -\frac{\partial H_0}{\partial p_j}, \quad j = 1,2,
\end{eqnarray}
generated from the Hamiltonian function
\begin{equation}
\label{Ham-fun}
H_0(p_1,q_1,p_2,q_2) = \frac{1}{4} (p_1^2 + q_1^2 + p_2^2 + q_2^2)^2 + \lambda_1 p_1 q_1 + \lambda_2 p_2 q_2.
\end{equation}
Since $H_0$ is $x$-independent, we introduce the constant $E_0 := 4 H_0(p_1,q_1,p_2,q_2)$ in addition
to parameters $(\lambda_1,\lambda_2)$. As is shown in \cite{Cao2}, the Hamiltonian system (\ref{Ham-sys})--(\ref{Ham-fun})
of degree two is integrable in the sense of Liouville and there exists another conserved quantity
\begin{eqnarray}
\nonumber
H_1(p_1,q_1,p_2,q_2) & = & 4 ( \lambda_1^3 p_1 q_1 + \lambda_2^3 p_2 q_2) - 4 (\lambda_1 p_1 q_1 + \lambda_2 p_2 q_2)^2
- (\lambda_1 (p_1^2 - q_1^2) + \lambda_2 (p_2^2 - q_2^2))^2\\
& \phantom{t} & + 2 (p_1^2 + q_1^2 + p_2^2 + q_2^2) (\lambda_1^2 (p_1^2 + q_1^2) + \lambda_2^2 (p_2^2 + q_2^2)).
\label{Ham-second}
\end{eqnarray}
Since $H_1$ is $x$-independent, we introduce another constant $E_1 := 4 H_1(p_1,q_1,p_2,q_2)$.
Thus, the algebraic method includes four parameters $(\lambda_1,\lambda_2,E_0,E_1)$.
It remains to establish differential equations on the class of admissible solutions $u$.

\subsection{Lax--Novikov equations}

By differentiating (\ref{potential}) in $x$ and using the Hamiltonian system (\ref{Ham-sys})--(\ref{Ham-fun}),
we obtain the following first-order differential equation:
\begin{equation}
\label{potential-DE}
\frac{du}{dx} = 2 \lambda_1 (p_1^2 - q_1^2) + 2 \lambda_2 (p_2^2 - q_2^2).
\end{equation}
It follows from (\ref{potential}) and (\ref{Ham-fun}) that
\begin{equation}
\label{potential-squared}
E_0 - u^2 =  4 \lambda_1 p_1 q_1 + 4 \lambda_2 p_2 q_2.
\end{equation}
By differentiating (\ref{potential-DE}) in $x$ and using the Hamiltonian system (\ref{Ham-sys})--(\ref{Ham-fun})
and the relation (\ref{potential-squared}), we obtain the following second-order differential equation:
\begin{eqnarray}
\nonumber
\frac{d^2u}{dx^2} + 2 u^3 & = & 2 E_0 u + 4 \lambda_1^2 (p_1^2 + q_1^2) + 4 \lambda_2^2 (p_2^2 + q_2^2) \\
& = & c u - 4 \lambda_2^2 (p_1^2 + q_1^2) - 4 \lambda_1^2 (p_2^2 + q_2^2),
\label{potential-DE2}
\end{eqnarray}
where we have introduced a parameter
\begin{equation}
\label{parameter-c}
c := 2 E_0 + 4 \lambda_1^2 + 4 \lambda_2^2.
\end{equation}
Taking yet another derivative of (\ref{potential-DE2}) in $x$ and using the Hamiltonian system (\ref{Ham-sys})--(\ref{Ham-fun}) again,
we obtain the following third-order differential equation:
\begin{eqnarray}
\frac{d^3u}{dx^3} + 6 u^2 \frac{du}{dx} = c \frac{d u}{dx} - 8 \lambda_1 \lambda_2 \left[ \lambda_2 (p_1^2 - q_1^2) + \lambda_1 (p_2^2 - q_2^2)\right].
\label{potential-DE3}
\end{eqnarray}
Differential equations (\ref{potential-DE}), (\ref{potential-DE2}), and (\ref{potential-DE3}) are not
closed for $u$. However, we show that one more differentiation gives a closed fourth-order differential equation on $u$.
To do so, we first note that it follows from (\ref{potential}), (\ref{Ham-fun}), (\ref{Ham-second}), (\ref{potential-DE}),
and (\ref{potential-DE2}) that
\begin{eqnarray}
\label{potential-squared-1}
E_1 = 16 ( \lambda_1^3 p_1 q_1 + \lambda_2^3 p_2 q_2) - (E_0 - u^2)^2
+ 2 u \left( \frac{d^2 u}{dx^2} + 2 u^3 - 2 E_0 u \right) - \left( \frac{du}{dx} \right)^2.
\end{eqnarray}
By using (\ref{potential-squared}), this relation can be rewritten in the equivalent form
\begin{eqnarray}
\label{potential-squared-2}
E_1 + E_0^2 - 4 E_0 (\lambda_1^2 + \lambda_2^2) + \left( \frac{du}{dx} \right)^2 - 2 u \frac{d^2 u}{d x^2} - 3 u^4 + c u^2 =
-16 \lambda_1\lambda_2 (\lambda_2 p_1 q_1 + \lambda_1 p_2 q_2).
\end{eqnarray}
By differentiating (\ref{potential-DE3}) in $x$ and using the Hamiltonian system (\ref{Ham-sys})--(\ref{Ham-fun})
and the relation (\ref{potential-squared-2}), we obtain a closed fourth-order differential equation on $u$:
\begin{equation}
\label{potential-DE4}
\frac{d^4 u}{dx^4} + 10 u^2 \frac{d^2 u}{dx^2} + 10 u \left( \frac{du}{dx} \right)^2
+ 6 u^5 = c\left( \frac{d^2 u}{dx^2} + 2 u^3 \right) + 2 d u,
\end{equation}
where we have introduced another parameter
\begin{equation}
\label{parameter-d}
d := E_1 + E_0^2 - 4 E_0 (\lambda_1^2 + \lambda_2^2) - 8 \lambda_1^2 \lambda_2^2.
\end{equation}
The fourth-order differential equation (\ref{potential-DE4}) belongs to the class of Lax--Novikov equations (see \cite{SV}
and references therein), which combine the stationary flows of the mKdV equation and higher-order members of its integrable hierarchy.

\subsection{Conserved quantities for the stationary fourth-order equation}

Let us introduce
\begin{equation}
\label{Q-W}
W(\lambda) = \left(\begin{array}{cc}
W_{11}(\lambda) & W_{12}(\lambda)\\
W_{12}(-\lambda) & -W_{11}(-\lambda) \end{array}\right),
\end{equation}
where the components of $W(\lambda)$ are given explicitly by
\begin{eqnarray}
\label{W11}
W_{11}(\lambda) & = & 1 - \frac{2 \lambda_1 p_1 q_1}{\lambda^2 - \lambda_1^2} - \frac{2 \lambda_2 p_2 q_2}{\lambda^2 - \lambda_2^2}, \\
\label{W12}
W_{12}(\lambda) & = & \frac{p_1^2}{\lambda-\lambda_1} + \frac{q_1^2}{\lambda + \lambda_1}
+ \frac{p_2^2}{\lambda-\lambda_2} + \frac{q_2^2}{\lambda + \lambda_2}.
\end{eqnarray}

In order to simplify the presentation, we denote derivatives of $u$ in $x$ by subscripts.
By using (\ref{potential}), (\ref{potential-DE}), (\ref{potential-DE2}), and (\ref{potential-DE3}),
the expression for $W_{12}$ can be rewritten in the equivalent form:
\begin{equation}
\label{W12-expansion}
W_{12}(\lambda) = \frac{\lambda^3 u + \frac{1}{2} \lambda^2 u_x + \frac{1}{4} \lambda (u_{xx} + 2 u^3 - cu)
+ \frac{1}{8} (u_{xxx} + 6 u^2 u_x - c u_x)}{(\lambda^2-\lambda_1^2)(\lambda^2 - \lambda_2^2)}.
\end{equation}
By using (\ref{potential-squared}), (\ref{potential-squared-2}), and (\ref{parameter-d}),
the expression for $W_{11}$ can also be rewritten in the equivalent form:
\begin{eqnarray}
\label{W11-expansion}
W_{11}(\lambda) = 1 - \frac{\lambda^2 (E_0 - u^2)}{2(\lambda^2-\lambda_1^2)(\lambda^2 - \lambda_2^2)}
- \frac{d + 8 \lambda_1^2 \lambda_2^2 + \left( u_x \right)^2 - 2 u u_{xx} - 3 u^4 + c u^2}{8(\lambda^2 - \lambda_1^2)(\lambda^2 - \lambda_2^2)}.
\end{eqnarray}
The Lax equation
\begin{equation}
\label{ham-Lax}
\frac{d}{dx} W(\lambda) = U(\lambda,u) W(\lambda) - W(\lambda) U(\lambda,u),
\end{equation}
is satisfied for every $\lambda \in \mathbb{C}$ if and only if
$(p_1,q_1,p_2,q_2)$ satisfy the Hamiltonian system (\ref{Ham-sys}) with (\ref{Ham-fun}),
where $u$ is represented by (\ref{potential}). In particular, the $(1,2)$-entry in the above relations yields the equation
\begin{equation}
\label{ham-1-2}
\frac{d}{dx} W_{12}(\lambda) = 2 \lambda W_{12}(\lambda) - 2 u W_{11}(\lambda).
\end{equation}
Substituting (\ref{W12-expansion}) and (\ref{W11-expansion}) into (\ref{ham-1-2}) yields the same fourth-order differential equation
(\ref{potential-DE4}).

It was shown in a similar context in \cite{Tu} that
$\det[W(\lambda)]$ has only simple poles at $\lambda = \pm \lambda_1$ and $\lambda = \pm \lambda_2$.
This can be confirmed from
(\ref{Q-W}), (\ref{W11}), and (\ref{W12}) with straightforward computations.
Substituting the representations (\ref{W12-expansion}) and (\ref{W11-expansion}) into $\det[W(\lambda)]$
and removing the double poles at $\lambda = \pm \lambda_1$ and $\lambda = \pm \lambda_2$ yield
the following two constraints:
\begin{eqnarray}
\nonumber
4 \lambda_j^2 \left[ u_{xx} + 2 u^3 - c u + 4 \lambda_j^2 u  \right]^2
- \left[ u_{xxx} + 6 u^2 u_x - c u_x + 4 \lambda_j^2 u_x \right]^2 & \phantom{t} & \\
- \left[ d + 8 \lambda_1^2 \lambda_2^2 + (u_x)^2 - 2 u u_{xx} - 3 u^4+ c u^2
+ 4 \lambda_j^2 (E_0 - u^2) \right]^2 = 0,& \phantom{t} & j = 1,2.
\label{constraints-DE}
\end{eqnarray}
These two equations (\ref{constraints-DE})
represent lower-order invariants for the fourth-order differential equation (\ref{potential-DE4}).
By performing elementary operations, the system of two constraints (\ref{constraints-DE})
can be rewritten in the equivalent form:
\begin{eqnarray}
\nonumber
\left( u_{xx} + 2 u^3 - 2 E_0 u \right)^2 - 16 \lambda_1^2 \lambda_2^2 u^2
- 2 u_x \left( u_{xxx} + 6 u^2 u_x - c u_x \right) & \phantom{t} & \\
\nonumber
- 2 (E_0 - u^2) \left[ d + 8 \lambda_1^2 \lambda_2^2 + (u_x)^2 - 2 u u_{xx} - 3 u^4 + c u^2  \right] & \phantom{t} &\\
- 4 (\lambda_1^2 + \lambda_2^2) \left[ (u_x)^2 + (E_0 - u^2)^2 \right] = 0 & \phantom{t} &
\label{constraints-DE-1}
\end{eqnarray}
and
\begin{eqnarray}
\nonumber
\left( u_{xxx} + 6 u^2 u_x - 2 E_0 u_x \right)^2 + \left( E_1 + E_0^2 + (u_x)^2 - 2 u u_{xx} - 3 u^4 + 2 E_0 u^2 \right)^2 & \phantom{t} & \\
\nonumber
- 4 (\lambda_1^2 + \lambda_2^2) (u_{xx} + 2 u^3 - 2 E_0 u)^2
- 16 \lambda_1^2 \lambda_2^2 \left[ (u_x)^2 + (E_0 - u^2)^2 \right] & \phantom{t} & \\
+ 32 \lambda_1^2 \lambda_2^2 u (u_{xx} + 2 u^3 - 2 E_0 u) = 0. & \phantom{t} &
\label{constraints-DE-2}
\end{eqnarray}
In order to verify conservation of (\ref{constraints-DE-1}) and (\ref{constraints-DE-2}), we have checked that
differentiating these constraints in $x$ recovers the fourth-order differential equation (\ref{potential-DE4}).

\subsection{Dubrovin equations}

One can characterize solutions $u$ to the fourth-order differential equation (\ref{potential-DE4})
from the algebro-geometric point of view, which is now commonly accepted
in the context of quasi-periodic solutions of integrable equations \cite{MatveevBook,Gestezy}. Since the numerator of $W_{12}(\lambda)$
is a polynomial of degree three, it admits three roots denoted as $\mu_1$, $\mu_2$, and $\mu_3$.
Writing the representation (\ref{W12-expansion}) in the form
\begin{equation}
\label{W-dubrovin}
W_{12}(\lambda) = \frac{u (\lambda - \mu_1)(\lambda - \mu_2) (\lambda - \mu_3)}{(\lambda^2 - \lambda_1^2)(\lambda^2 - \lambda_2^2)},
\end{equation}
yields the following differential expressions for $\mu_1$, $\mu_2$, and $\mu_3$:
\begin{eqnarray}
\label{mu1}
\mu_1 + \mu_2 + \mu_3 & = &  -\frac{1}{2u} \frac{du}{dx}, \\
\label{mu2}
\mu_1\mu_2 + \mu_1 \mu_3 + \mu_2 \mu_3 & = &  \frac{1}{4u} \left[ \frac{d^2 u}{dx^2} + 2 u^3 - c u \right], \\
\label{mu3}
\mu_1 \mu_2 \mu_3 & = &  -\frac{1}{8u} \left[ \frac{d^3 u}{dx^3} + 6 u^2 \frac{du}{dx} - c \frac{du}{dx} \right].
\end{eqnarray}

By substituting the representations (\ref{W12-expansion}) and (\ref{W11-expansion}) into
$\det\left[W(\lambda)\right]$, removing the double poles with the constraints (\ref{constraints-DE}),
and simplifying the remaining expressions, we obtain the following representation:
\begin{equation}
\label{det-W-expansion}
\det[W(\lambda)] = -1 + \frac{4 E_0 (\lambda^2 - \lambda_1^2 - \lambda_2^2) + E_1}{4 (\lambda^2 - \lambda_1^2) (\lambda^2 - \lambda_2^2)}
= -\frac{b(\lambda)}{a(\lambda)},
\end{equation}
where
\begin{eqnarray*}
\left\{ \begin{array}{l}
a(\lambda) := (\lambda^2 - \lambda_1^2) (\lambda^2 - \lambda_2^2),\\
b(\lambda) := (\lambda^2 - \lambda_1^2) (\lambda^2 - \lambda_2^2) - E_0 (\lambda^2 - \lambda_1^2 - \lambda_2^2) - \frac{1}{4} E_1. \end{array} \right.
\end{eqnarray*}

By substituting (\ref{W-dubrovin}) and (\ref{det-W-expansion}) into (\ref{ham-1-2}) and evaluating (\ref{ham-1-2})
at $\lambda = \mu_{1,2,3}$, we recover the Dubrovin equations \cite{Dubrovin} in the form:
\begin{equation}
\label{Dubrovin-eq}
\frac{d \mu_j}{d x} = \frac{2 \sqrt{a(\mu_j) b(\mu_j)}}{\prod_{i \neq j} (\mu_j - \mu_i)}.
\end{equation}
These equations characterize the algebro-geometric structure of the quasi-periodic solutions.
The variables $\{ \mu_1, \mu_2, \mu_3 \}$ are called Dubrovin's variables for the quasi-periodic solutions.
The presence of three Dubrovin variables indicates that the general solution to the fourth-order equation (\ref{potential-DE4})
is expressed by the Riemann Theta function of {\em genus three} \cite{MatveevBook,Gestezy}.

\subsection{Degeneration procedure}

For the scopes of our work, we reduce the algebraic method in order
to recover the third-order differential equation (\ref{third-order-intro}) instead of
the fourth-order equation (\ref{potential-DE4}). This scope is achieved by imposing
the constraint $\mu_3 = 0$ on Dubrovin variables satisfying the system of equations
(\ref{mu1}), (\ref{mu2}), and (\ref{mu3}). Indeed, if $\mu_3 = 0$, equation (\ref{mu3}) is equivalent
to the third-order differential equation (\ref{third-order-intro}) rewritten again as
\begin{eqnarray}
\label{third-order}
\frac{d^3 u}{dx^3} + 6 u^2 \frac{du}{dx} - c \frac{du}{dx} = 0,
\end{eqnarray}
whereas the other two equations (\ref{mu1}) and (\ref{mu2}) yield relations
\begin{eqnarray}
\label{mu1-again}
\mu_1 + \mu_2 & = &  -\frac{1}{2u} \frac{du}{dx}, \\
\label{mu2-again}
\mu_1\mu_2 & = &  \frac{1}{4u} \left[ \frac{d^2 u}{dx^2} + 2 u^3 - c u \right].
\end{eqnarray}
The presence of two Dubrovin variables indicates that the general solution to the third-order equation (\ref{third-order})
is expressed by the Riemann Theta function of {\em genus two} \cite{MatveevBook,Gestezy}.

Substituting (\ref{third-order}) into (\ref{potential-DE4}) yields the following invariant
for the third-order equation (\ref{third-order}):
\begin{eqnarray}
\label{first-order}
\left( \frac{du}{dx} \right)^2 - 2 u \frac{d^2 u}{dx^2} - 3 u^4 + c u^2 + d = 0.
\end{eqnarray}
Indeed, differentiating (\ref{first-order}) in $x$ recovers the third-order equation (\ref{third-order}).

Note in passing that the representations (\ref{W12-expansion}) and (\ref{W11-expansion})
are factorized by one power of $\lambda$, after equations (\ref{third-order}) and (\ref{first-order})
are used, namely:
$$
\lambda^{-1} W_{12}(\lambda) = \frac{\lambda^2 u + \frac{1}{2} \lambda u_x + \frac{1}{4} (u_{xx} + 2 u^3 - c u)}{
(\lambda^2 - \lambda_1^2)(\lambda^2 - \lambda_2^2)}
$$
and
$$
\lambda^{-1} W_{11}(\lambda) = \frac{\lambda^3 - \frac{1}{2} \lambda (E_0 + 2 \lambda_1^2 + 2 \lambda_2^2 - u^2)}{
(\lambda^2 - \lambda_1^2)(\lambda^2 - \lambda_2^2)}.
$$
These representations coincide with those used in \cite{Wright}
for integration of ultra-elliptic solutions to the cubic NLS equation.

\begin{remark}
It may be possible to modify the algebraic method by starting with
polynomials of even degree for $W_{12}(\lambda)$ and odd degree for $W_{11}(\lambda)$
and to derive the same solutions as those obtained by the degeneration procedure here.
This possible modification remains to be developed.
\end{remark}

By substituting differential equations (\ref{third-order}) and (\ref{first-order}) into the constraints
(\ref{constraints-DE-1}) and (\ref{constraints-DE-2}), we rewrite the constraints in the equivalent form:
\begin{eqnarray}
\label{constraints-DE-1a}
\left( \frac{d^2 u}{d x^2} + 2 u^3 - 2 E_0 u \right)^2 - 4 (\lambda_1^2 + \lambda_2^2)
\left[ \left( \frac{du}{dx} \right)^2 + (E_0 - u^2)^2 \right] - 16 \lambda_1^2 \lambda_2^2 E_0 = 0
\end{eqnarray}
and
\begin{eqnarray}
\nonumber
4 (\lambda_1^4 + \lambda_1^2 \lambda_2^2 + \lambda_2^4)
\left[ \left( \frac{du}{dx} \right)^2 +(E_0 - u^2)^2 \right]
+ 16 \lambda_1^2 \lambda_2^2 (\lambda_1^2 + \lambda_2^2) (E_0 - u^2) + 16 \lambda_1^4 \lambda_2^4 \\
- (\lambda_1^2 + \lambda_2^2) \left( \frac{d^2 u}{dx^2} + 2u^3 - 2 E_0 u \right)^2
+ 8 \lambda_1^2 \lambda_2^2 u \left( \frac{d^2 u}{dx^2} + 2 u^3 - 2 E_0 u \right) = 0.
\label{constraints-DE-2a}
\end{eqnarray}

In order to simplify characterization of the general solution of the third-order equation (\ref{third-order}),
we integrate it once and obtain the differential equation (\ref{second-order-intro}) rewritten again as
\begin{eqnarray}
\label{second-order}
\frac{d^2 u}{dx^2} + 2 u^3 - c u = e,
\end{eqnarray}
where $e$ is the integration constant. Substituting (\ref{second-order}) into (\ref{first-order}) yields
the first-order invariant (\ref{first-order-intro}) for the second-order equation (\ref{second-order})
rewritten again as
\begin{eqnarray}
\label{zero-order}
\left( \frac{du}{dx} \right)^2 + u^4 - c u^2 + d = 2 e u.
\end{eqnarray}
Indeed, differentiating (\ref{zero-order}) in $x$ recovers the second-order equation (\ref{second-order}).

Substituting (\ref{second-order}) and (\ref{zero-order}) into (\ref{constraints-DE-1a})
and (\ref{constraints-DE-2a}) yields the following two constraints on the parameters
of the algebraic method:
\begin{equation}
e^2 = 4 (\lambda_1^2 + \lambda_2^2) (E_0^2 - d) + 16 \lambda_1^2 \lambda_2^2 E_0
\label{constraints-DE-3}
\end{equation}
and
\begin{equation}
4 (\lambda_1^4 + \lambda_1^2 \lambda_2^2 + \lambda_2^4) (E_0^2 - d )
+ 16 \lambda_1^2 \lambda_2^2 (\lambda_1^2 + \lambda_2^2) E_0 + 16 \lambda_1^4 \lambda_2^4
- (\lambda_1^2 + \lambda_2^2) e^2 = 0.
\label{constraints-DE-4}
\end{equation}
The system (\ref{constraints-DE-3}) and (\ref{constraints-DE-4}) is solved in the explicit form:
\begin{equation}
\label{parameter-de}
d = E_0^2 - 4 \lambda_1^2 \lambda_2^2
\end{equation}
and
\begin{equation}
\label{parameter-e}
e^2 = 16 \lambda_1^2 \lambda_2^2 (E_0 + \lambda_1^2 + \lambda_2^2).
\end{equation}
The relation (\ref{parameter-de}) and the definition (\ref{parameter-d}) imply that
the parameter $E_1$ is uniquely expressed by
\begin{equation}
E_1 = 4 E_0 (\lambda_1^2 + \lambda_2^2) + 4 \lambda_1^2 \lambda_2^2.
\label{parameter-E-1}
\end{equation}

Summarizing, the three parameters $(c,d,e)$ for the differential equations (\ref{third-order}),
(\ref{second-order}), and (\ref{zero-order}) are related to the parameters
$(\lambda_1,\lambda_2,E_0)$ of the algebraic method
with the help of relation (\ref{parameter-c}), (\ref{parameter-de}), and (\ref{parameter-e}),
whereas parameter $E_1$ is uniquely expressed by (\ref{parameter-E-1}).

\begin{remark}
Compared to the algebraic method with only one eigenfunction $\varphi = (p_1,q_1)^t$ in \cite{CPkdv},
where solutions to the second-order equation (\ref{second-order}) with $e = 0$ are obtained,
we have flexibility to recover the same solutions with the two eigenfunctions
$\varphi = (p_1,q_1)^t$ and $\varphi = (p_2,q_2)^t$ as in (\ref{potential}), as long as there are two nonzero
eigenvalues for $\lambda_1$ and $\lambda_2$ such that $\lambda_1 \neq \pm \lambda_2$.
Thus, solutions to the second-order equation (\ref{second-order}) with $e = 0$
can be recovered by two equivalent algebraic methods.
\end{remark}

\section{Proof of Theorem \ref{theorem-wave}}
\label{sec-3}

Here we develop the phase-plane analysis of the second-order differential equation (\ref{second-order-intro}).
Therefore, we treat it as a dynamical system on the phase plane $(u,u')$, where $u' := \frac{du}{dx}$.
If $c > 0$ and $e \in (-e_0,e_0)$, where $e_0 := 2 \sqrt{c^3}/(3 \sqrt{6})$, the cubic polynomial
\begin{equation}
\label{poly-Q}
Q(u) := 2 u^3 - c u - e,
\end{equation}
has three real roots ordered as $u_* < u_{**} < u_{***}$.
Since $Q'(u) = 6 u^2 - c$, we have $Q'(u_*) > 0$, $Q'(u_{**}) < 0$, and
$Q'(u_{***}) > 0$, hence $(u_*,0)$ and $(u_{***},0)$ are center
points on the phase plane, whereas $(u_{**},0)$ is a saddle point.
Trajectories on the phase plane coincide with the level curves
of the energy function
$$
\mathcal{H}(u,u') = (u')^2 + u^4 - c u^2 - 2 e u,
$$
which are given by the constant value $-d$ in the first-order invariant (\ref{first-order-intro}).
The level curves of $\mathcal{H}(u,u')$ in the case $c > 0$ and $e \in (-e_0,e_0)$ are shown on Figure \ref{fig-plane}.

\begin{figure}[ht]
\centering
\includegraphics[scale=0.5]{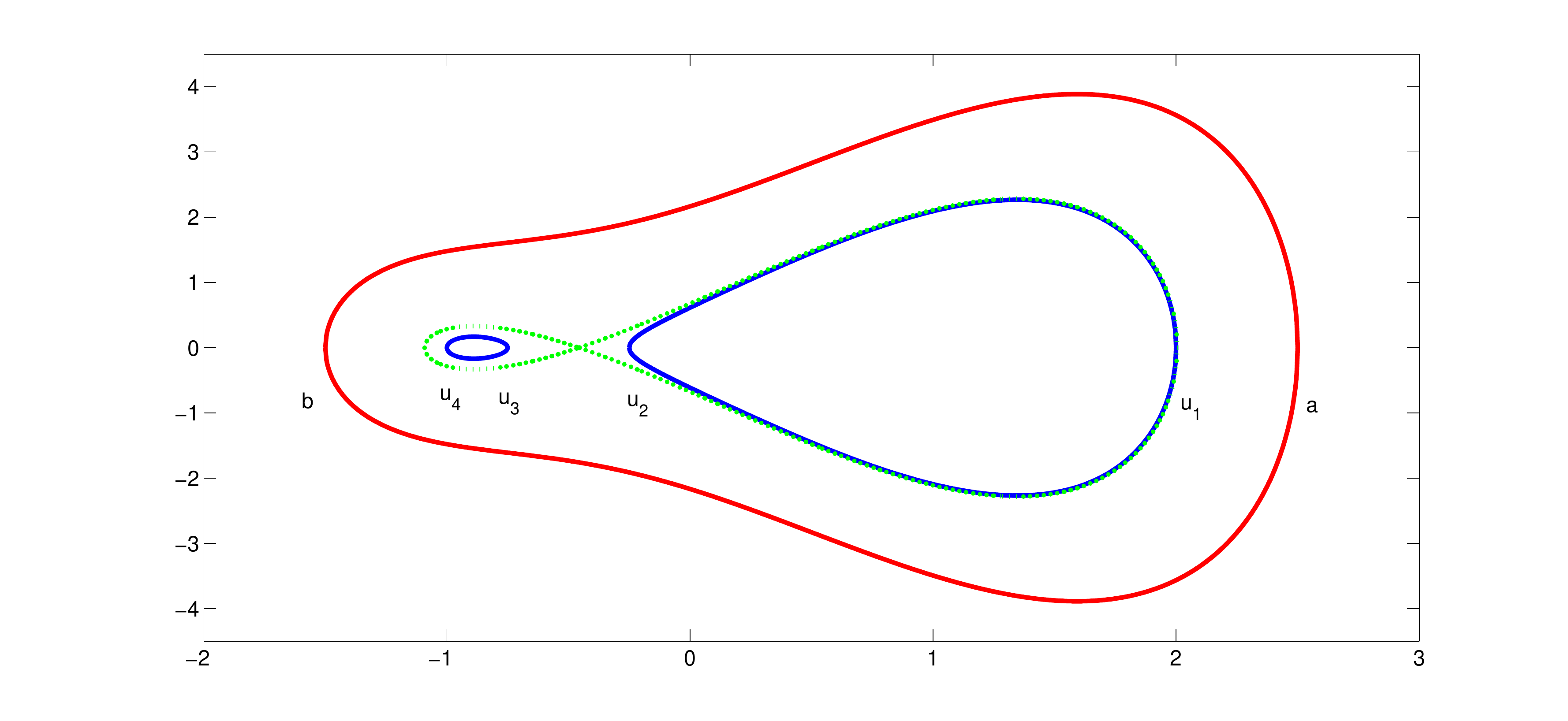}
\caption{Phase plane $(u,u')$ for the second-order equation (\ref{second-order-intro}). } \label{fig-plane}
\end{figure}

Since $2 Q(u) = \partial_u H(u,0)$, the function $H(u,0)$ has three extremal points with two minima
at $u_*$ and $u_{***}$ and one maximum at $u_{**}$. Let $h_* := \mathcal{H}(u_*,0)$,
$h_{**} := \mathcal{H}(u_{**},0)$, and $h_{***} := \mathcal{H}(u_{***},0)$.
Recall that $\mathcal{H}(u,0) = P(u) - d$, where $P(u)$ is given by (\ref{polynomial}).
For every $d \in (d_1,d_2)$ with $d_1 = -h_{**}$ and $d_2 = -\max\{h_*,h_{***}\}$,
the polynomial $P(u)$ has four real roots ordered as $u_4 < u_3 < u_2 < u_1$ also
shown on Figure \ref{fig-plane}. There exists two periodic solutions
which correspond to closed orbits on the phase plane: one closed orbit
surrounds the center point $(u_*,0)$ and the other closed orbit surrounds
the center point $(u_{***},0)$. These periodic solutions are shown
on Figure \ref{fig-solution}(a,b).

For every $d \in (-\infty,d_1) \cup (d_2,d_3)$, where $d_3 := -\min\{h_*,h_{***}\}$,
the polynomial $P(u)$ has only two real roots ordered as $b \leq a$ also shown on Figure \ref{fig-plane}.
There exists only one periodic solution
which corresponds to either the closed orbit surrounding all three critical points
$(u_*,0)$, $(u_{**},0)$, and $(u_{***},0)$ if $d \in (-\infty,d_1)$
or the closed orbit surrounding only one center point if $d \in (d_2,d_3)$.
Note that $d_2 < d_3$ if $e \neq 0$, whereas $d_2 = d_3$ if $e = 0$.
The periodic solution surrounding all three critical points is shown on Figure \ref{fig-solution}(c).

If either $c \leq 0$ or $c > 0$ and $e \in (-\infty,-e_0] \cup [e_0,\infty)$,
the cubic polynomial $Q(u)$ in (\ref{poly-Q}) admits only one real root labeled as $u_*$,
which corresponds to the global minimum of $\mathcal{H}(u,0)$.
For every $d \in (-\infty, d_1)$, where $d_1 := -\mathcal{H}(u_*,0)$,
there exists only one periodic solution which corresponds to the closed orbit surrounding
the only center point $(u_*,0)$ of the dynamical system.

\begin{figure}[ht]
\includegraphics[scale=0.3]{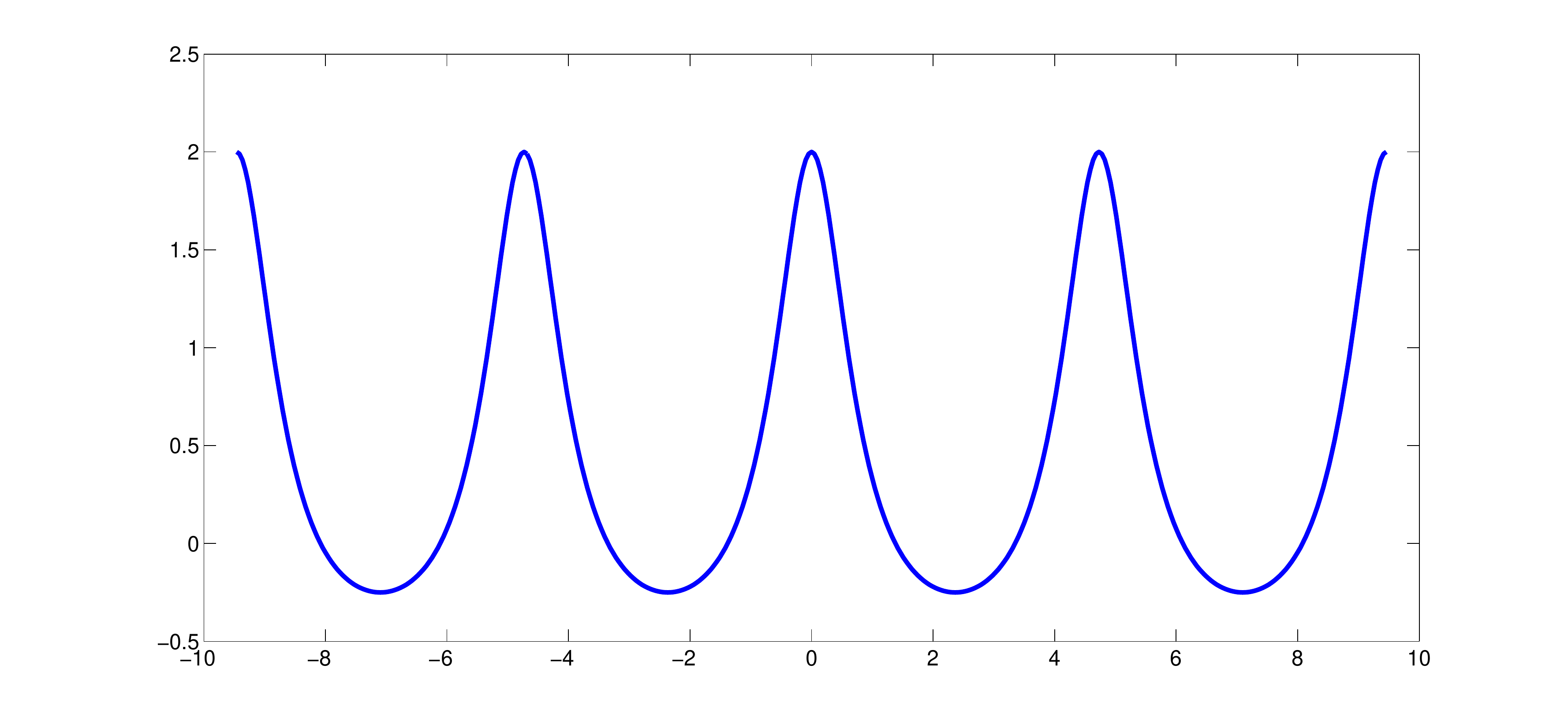}\\
\includegraphics[scale=0.3]{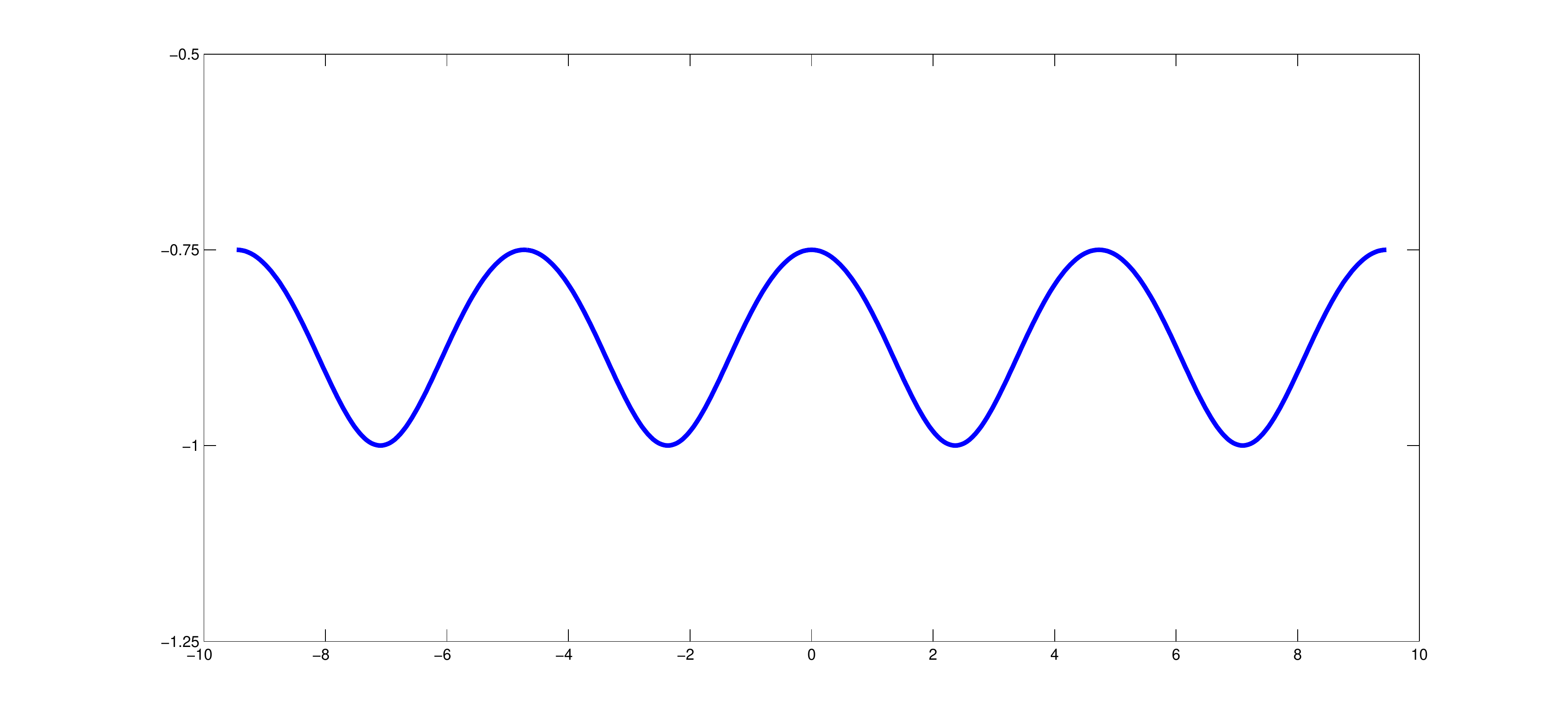}\\
\includegraphics[scale=0.3]{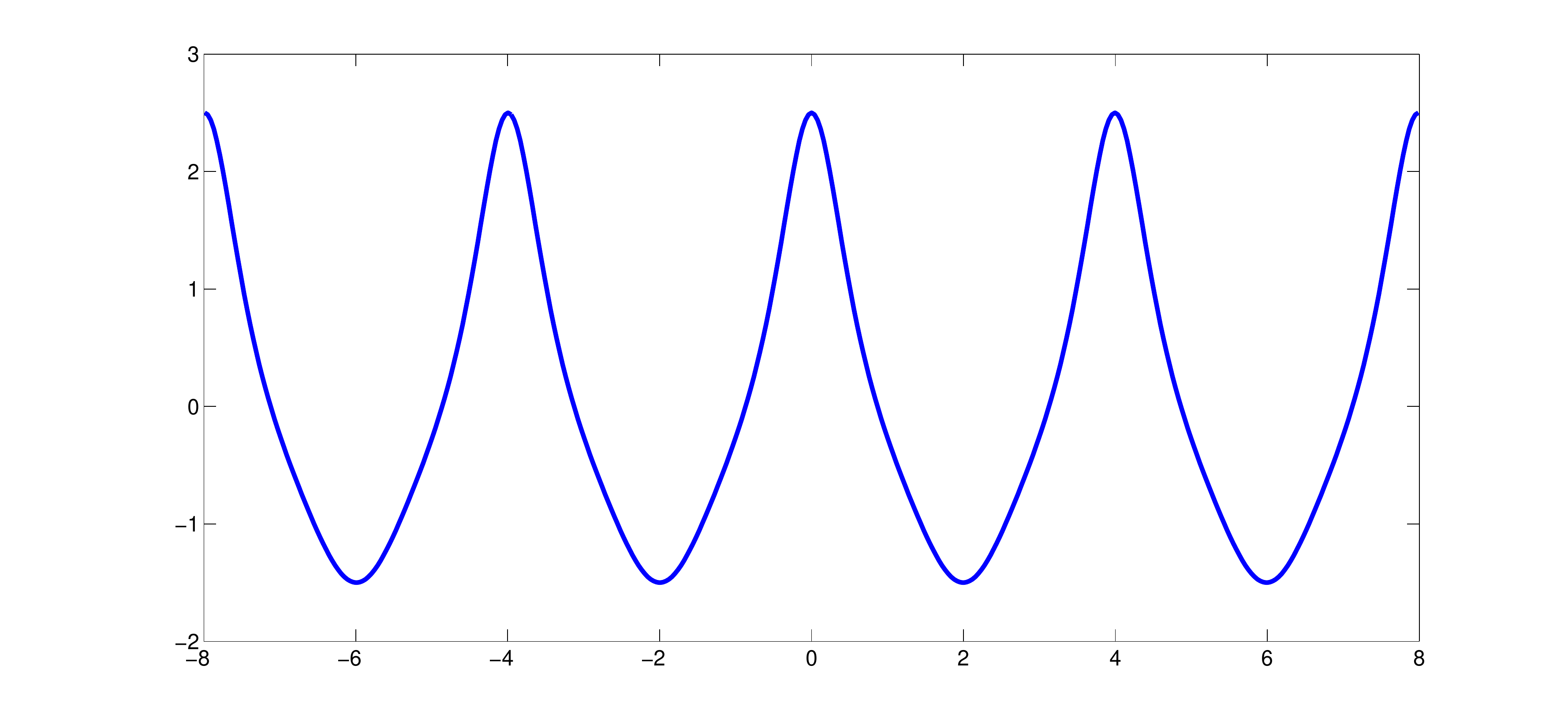}
\caption{Three periodic solutions for the phase plane on Figure \ref{fig-plane}:
(a) inside $[u_2,u_1]$, (b) inside $[u_4,u_3]$, (c) inside $[b,a]$. } \label{fig-solution}
\end{figure}

In order to complete the proof of Theorem \ref{theorem-wave},
it remains to verify the exact representations
(\ref{Jacob-1-intro})--(\ref{roots-1-intro})
and (\ref{Jacob-2-intro})--(\ref{roots-2-intro})
of the periodic solutions as
rational functions of the Jacobian elliptic functions.
The following two lemmas state the corresponding exact representations.

\begin{lemma}
\label{lemma-dn}
Let $P(u)$ in (\ref{polynomial}) admit four real roots ordered as
$u_4 \leq u_3 \leq u_2 \leq u_1$. The exact periodic solution
to the first-order invariant (\ref{first-order-intro}) in the interval
$[u_2,u_1]$ is given by
\begin{equation}
u(x) = u_4 + \frac{(u_1-u_4) (u_2-u_4)}{(u_2-u_4) + (u_1 - u_2) {\rm sn}^2(\nu x;\kappa)},
\label{Jacob-1}
\end{equation}
where $\nu > 0$ and $\kappa \in (0,1)$ are parameters given by
\begin{equation}
\label{roots-1}
\left\{ \begin{array}{l}
4 \nu^2 = (u_1 - u_3) (u_2 - u_4), \\
4 \nu^2 \kappa^2 = (u_1 - u_2) (u_3 - u_4). \end{array} \right.
\end{equation}
The roots satisfy the constraint
$u_1 + u_2 + u_3 + u_4 = 0$ and are related to the parameters
$(c,d,e)$ by
\begin{eqnarray}
\left\{ \begin{array}{l}
c = -(u_1 u_2 + u_1 u_3 + u_1 u_4 + u_2 u_3 + u_2 u_4 + u_3 u_4), \\
2e = u_1 u_2 u_3 + u_1 u_2 u_4 + u_1 u_3 u_4 + u_2 u_3 u_4, \\
d = u_1 u_2 u_3 u_4. \end{array} \right.
\label{u-roots}
\end{eqnarray}
\end{lemma}

\begin{proof}
By factorizing $P(u)$ by its roots, we rewrite
the first-order invariant (\ref{first-order-intro}) in the form:
\begin{eqnarray}
\label{factorized}
\left( \frac{du}{dx} \right)^2 + (u-u_1) (u-u_2) (u-u_3) (u-u_4) = 0.
\end{eqnarray}
Expanding the quartic polynomial and comparing it with the coefficients in (\ref{first-order-intro}),
we verify the constraint $u_1 + u_2 + u_3 + u_4 = 0$ and the relations (\ref{u-roots}).
In order to prove that (\ref{Jacob-1}) with (\ref{roots-1}) is an exact solution of (\ref{factorized}),
we substitute $u(x) = u_4 + (u_1-u_4) (u_2-u_4)/v(x)$
into (\ref{factorized}) and obtain the equivalent equation:
\begin{eqnarray}
\label{factorized-11}
\left( \frac{dv}{dx} \right)^2 + \left[(u_1 - u_4) - v\right]
\left[ (u_2-u_4) - v \right] \left[ (u_1-u_4)(u_2-u_4) - (u_3 - u_4) v \right] = 0.
\end{eqnarray}
Then, we substitute $v(x) = (u_2-u_4) + (u_1 - u_2) w(x)$ into (\ref{factorized-11})
and obtain the equivalent equation
\begin{eqnarray}
\label{factorized-12}
\left( \frac{dw}{dx} \right)^2 = w (1-w) \left[ (u_1-u_3)(u_2-u_4) - (u_1 - u_2)(u_3-u_4) w \right].
\end{eqnarray}
Taking derivative of $w(x) = {\rm sn}^2(\nu x;\kappa)$ and defining
$\nu$ and $\kappa$ by (\ref{roots-1}) satisfies (\ref{factorized-12}).
Since the Jacobian elliptic function ${\rm sn}^2(\nu x;\kappa)$ is periodic in $x$ with
the period $L := 2K(\kappa)/\nu$, where $K(\kappa)$ is the complete elliptic integral,
the solution $u$ in (\ref{Jacob-1}) is $L$-periodic. It is unique up
to the translation symmetry $u(x) \mapsto u(x+x_0)$, $x_0 \in \mathbb{R}$.
Furthermore, it is defined in the interval $[u_2,u_1]$ with $u(0) = u_1$ and $u(L/2) = u_2$.
\end{proof}

\begin{remark}
As follows from Figure \ref{fig-plane}, there exists another periodic solution
in the interval $[u_4,u_3]$ for the same combination of roots $\{ u_1,u_2,u_3,u_4\}$.
This solution is obtained by the symmetry transformation (S2) in (\ref{symm-intro}).
The corresponding exact periodic solution is
\begin{equation}
u(x) = u_2 + \frac{(u_3-u_2) (u_4-u_2)}{(u_4-u_2) + (u_3 - u_4) {\rm sn}^2(\nu x;\kappa)},
\label{Jacob-1-symm}
\end{equation}
whereas the relations (\ref{roots-1}) remain the same. The new solution (\ref{Jacob-1-symm})
is defined in the interval $[u_4,u_3]$ with $u(0) = u_3$ and $u(L/2) = u_4$.
The new solution (\ref{Jacob-1-symm}) is also $L$-periodic with the same
period $L = 2 K(\kappa)/\nu$.
\end{remark}

\begin{remark}
The symmetry transformation (S1) in (\ref{symm-intro}) generates another solution
\begin{equation}
u(x) = u_3 + \frac{(u_2-u_3) (u_1-u_3)}{(u_1-u_3) + (u_2 - u_1) {\rm sn}^2(\nu x;\kappa)},
\label{Jacob-1-symm-a}
\end{equation}
which also exists in the interval $[u_2,u_1]$ with $u(0) = u_2$ and $u(L/2) = u_1$.
The symmetry transformation (S3) in (\ref{symm-intro}) generates another solution
\begin{equation}
u(x) = u_1 + \frac{(u_4-u_1) (u_3-u_1)}{(u_3-u_1) + (u_4 - u_3) {\rm sn}^2(\nu x;\kappa)},
\label{Jacob-1-symm-b}
\end{equation}
which also exists in the interval $[u_4,u_3]$ with $u(0) = u_4$ and $u(L/2) = u_3$.
By uniqueness of the closed orbits for periodic solutions on the phase plane,
the solution (\ref{Jacob-1-symm-a}) coincides with the periodic solution (\ref{Jacob-1})
translated by the half-period $L/2$, whereas the solution (\ref{Jacob-1-symm-b}) coincides
with the periodic solution (\ref{Jacob-1-symm}) also translated by the half-period $L/2$.
\end{remark}

\begin{example}
\label{example-1}
Let $e = 0$. Since $P(u)$ is even in (\ref{polynomial}) for $e = 0$, we have $u_4 = -u_1$ and $u_3 = -u_2$. It follows
from (\ref{u-roots}) that $c = u_1^2 + u_2^2$ and $d = u_1^2 u_2^2$, whereas
relations (\ref{roots-1}) yield $u_1 = \nu (1 + \kappa)$ and $u_2 = \nu (1 - \kappa)$, from which
the exact solution (\ref{Jacob-1}) takes the equivalent form
\begin{equation}
u(x) = -u_1 + \frac{2 u_1}{1 + \kappa {\rm sn}^2(\nu x;\kappa)} = u_1
\frac{1 - \kappa {\rm sn}^2(\nu x;\kappa)}{1 + \kappa {\rm sn}^2(\nu x;\kappa)}.
\label{Jacob-1-simple}
\end{equation}
From Table 8.152 of \cite{Grad}, we have the transformation formula
$$
\frac{1 - \kappa {\rm sn}^2(\nu x;\kappa)}{1 + \kappa {\rm sn}^2(\nu x;\kappa)} =
{\rm dn}(\nu (1+\kappa) x; k), \quad k := \frac{2 \sqrt{\kappa}}{1 + \kappa},
$$
from which we derive
\begin{equation}
\label{dn-solution}
u(x) = u_1 {\rm dn}(u_1 x;k), \quad k = \sqrt{1-\frac{u_2^2}{u_1^2}}.
\end{equation}
Setting $u_1 = 1$ and $u_2 = \sqrt{1 - k^2}$ yields
the normalized dnoidal periodic wave (\ref{dn-intro}).
\end{example}

\begin{example}
\label{example-2}
Let $u_1 = u_2$. The exact solution (\ref{Jacob-1}) becomes the constant solution
$u(x) = u_1$.
\end{example}

\begin{example}
\label{example-3}
Let $u_2 = u_3$. It follows from (\ref{roots-1}) that $\kappa = 1$
so that the exact solution (\ref{Jacob-1}) becomes the exponential soliton
on the constant background $u_2$:
\begin{equation}
\label{Akh-solitary}
u(x) = u_4 + \frac{(u_1 - u_4) (u_2 - u_4)}{(u_2 - u_4) + (u_1 - u_2) \tanh^2(\nu x)}, \quad \nu = \frac{1}{2} \sqrt{(u_1-u_2)(u_2-u_4)},
\end{equation}
where roots satisfy the constraint $u_1 + 2 u_2 + u_4 = 0$.
\end{example}

\begin{example}
\label{example-4}
Let $u_3 = u_4$. It follows from (\ref{roots-1}) that $\kappa = 0$
so that the exact solution (\ref{Jacob-1}) becomes
$$
u(x) = u_4 + \frac{(u_1-u_4)(u_2-u_4)}{(u_2-u_4) + (u_1-u_2) \sin^2(\nu x)}, \quad \nu = \frac{1}{2} \sqrt{(u_1-u_4)(u_2-u_4)},
$$
where roots satisfy the constraint $u_1 + u_2 + 2 u_4 = 0$.
This periodic wave solution is obtained in \cite{Nail1} and reviewed recently in \cite{Choudhury}.
Setting $u_1 = 1+2b$, $u_2 = 1 - 2b$, and $u_4 = -1$ with $b \in [0,1]$
yields $\nu = \sqrt{1 - b^2}$ and transform the explicit solution to the form:
\begin{equation}
\label{Akh-periodic}
u(x) = -1 + \frac{2(1-b^2)}{1-b \cos(2 \sqrt{1-b^2}x)}.
\end{equation}
It follows from (\ref{u-roots}) that $c = 2 + 4b^2$, $d = 1 - 4b^2$, and $e = 4 b^2$.
As $b \to 1$, the explicit solution (\ref{Akh-periodic})
converges to the algebraic soliton (\ref{alg-soliton}) at $t = 0$
and after reflection $u \mapsto -u$.
\end{example}

\begin{lemma}
\label{lemma-cn}
Let $P(u)$ in (\ref{polynomial}) admit two real roots ordered as
$b \leq a$ and two complex conjugated roots labeled as $\alpha \pm i \beta$.
The exact periodic solution to the first-order invariant (\ref{first-order-intro}) is given by
\begin{equation}
u(x) = a + \frac{(b-a) (1 - {\rm cn}(\nu x; \kappa))}{1 + \delta + (\delta - 1) {\rm cn}(\nu x;\kappa)},
\label{Jacob-2}
\end{equation}
where $\delta > 0$, $\nu > 0$, and $\kappa \in (0,1)$ are parameters given by
\begin{equation}
\label{roots-2}
\left\{ \begin{array}{l}
\delta^2 = \frac{(b-\alpha)^2 + \beta^2}{(a-\alpha)^2 + \beta^2}, \\
\nu^2 = \sqrt{\left[ (a-\alpha)^2 + \beta^2 \right] \left[ (b-\alpha)^2 + \beta^2 \right]}, \\
2 \kappa^2 = 1 - \frac{(a-\alpha) (b - \alpha) + \beta^2}{\sqrt{\left[ (a-\alpha)^2 + \beta^2 \right] \left[ (b-\alpha)^2 + \beta^2 \right]}}. \end{array} \right.
\end{equation}
The roots satisfy the constraint $a + b + 2 \alpha = 0$ and are related to the parameters
$(c,d,e)$ by
\begin{eqnarray}
\left\{ \begin{array}{l}
c = -(ab + 2 \alpha (a + b) + \alpha^2 + \beta^2), \\
2e = 2 \alpha ab + (a+b) (\alpha^2 + \beta^2), \\
d = ab (\alpha^2 + \beta^2). \end{array} \right.
\label{u-roots-complex}
\end{eqnarray}
\end{lemma}

\begin{proof}
By factorizing $P(u)$ by its roots, we rewrite
the first-order invariant (\ref{first-order-intro}) in the form:
\begin{eqnarray}
\label{factorized-complex}
\left( \frac{du}{dx} \right)^2 + (u-a) (u-b) \left[ (u-\alpha)^2 + \beta^2 \right]  = 0.
\end{eqnarray}
Expanding the quartic polynomial and comparing it with the coefficients in (\ref{first-order-intro}),
we verify the constraint $a + b + 2 \alpha = 0$ and the relations (\ref{u-roots-complex}).
In order to prove that (\ref{Jacob-2}) with (\ref{roots-2}) is an exact solution of (\ref{factorized-complex}),
we substitute $u(x) = a + (b-a)/v(x)$ into (\ref{factorized-complex}) and obtain the equivalent equation:
\begin{eqnarray}
\label{factorized-15}
\left( \frac{dv}{dx} \right)^2 + (1-v)
\left[ ((a-\alpha) v + b-a)^2 + \beta^2 v^2 \right] = 0.
\end{eqnarray}
Then, we substitute
$$
v(x) = 1 + \delta \frac{1 + {\rm cn}(\nu x;\kappa)}{1 - {\rm cn}(\nu x;\kappa)}
$$
into (\ref{factorized-15}) and obtain three equations on $\delta$, $\nu$, and $\kappa$ which yield (\ref{roots-2}).
Since the Jacobian elliptic function ${\rm cn}(\nu x;\kappa)$ is periodic in $x$ with
the period $L := 4K(\kappa)/\nu$, where $K(\kappa)$ is the complete elliptic integral,
the solution $u$ in (\ref{Jacob-2}) is $L$-periodic. It is unique up
to the translation symmetry $u(x) \mapsto u(x+x_0)$, $x_0 \in \mathbb{R}$.
Furthermore, it is defined in the interval $[b,a]$ with $u(0) = a$ and $u(L/2) = b$.
\end{proof}

\begin{remark}
The symmetry transformation (S0) in (\ref{symm-4-intro}) generates another solution:
\begin{equation}
u(x) = b + \frac{(a-b) (1 - {\rm cn}(\nu x; \kappa))}{1 + \delta^{-1} + (\delta^{-1} - 1) {\rm cn}(\nu x;\kappa)},
\label{Jacob-2-symm}
\end{equation}
which also exists in the interval $[b,a]$ with $u(0) = b$ and $u(L/2) = a$.
By uniqueness of the closed orbits for periodic solutions on the phase plane,
the solution (\ref{Jacob-2-symm}) coincides with the periodic solution (\ref{Jacob-2})
translated by the half-period $L/2$.
\end{remark}

\begin{example}
\label{example-5}
Let $e = 0$. Since $P(u)$ is even in (\ref{polynomial}) for $e = 0$, we have $b = -a$, $\alpha = 0$, and $\delta = 1$.
The exact solution (\ref{Jacob-2}) becomes
$$
u(x) = a {\rm cn}(\nu x;\kappa), \quad \nu = \sqrt{a^2 + \beta^2}, \quad \kappa = \frac{a}{\sqrt{a^2 + \beta^2}}.
$$
Setting $a = k$ and $\beta = \sqrt{1 - k^2}$ yields
the normalized cnoidal periodic wave (\ref{cn-intro}).
\end{example}

\section{Proof of Theorem \ref{theorem-eig}}
\label{sec-4}

Combining (\ref{parameter-c}), (\ref{parameter-de}), and (\ref{parameter-e}) yields
the following system of algebraic equations
\begin{equation}
\label{system-parameters}
\left\{ \begin{array}{l} c = 2 E_0 + 4 (\lambda_1^2 + \lambda_2^2), \\
d = E_0^2 - 4 \lambda_1^2 \lambda_2^2, \\
e^2 = 16 \lambda_1^2 \lambda_2^2 (E_0 + \lambda_1^2 + \lambda_2^2). \end{array} \right.
\end{equation}
Let us introduce $y := 4(\lambda_1^2 + \lambda_2^2)$ and $z := 4 \lambda_1^2 \lambda_2^2$.
Substituting $y = c - 2E_0$ and $z = E_0^2 - d$ into the third equation of the system (\ref{system-parameters})
yields the cubic equation for $E_0$:
\begin{equation}
\label{cubic-eq-1}
(E_0^2 - d) (2E_0 + c) = e^2.
\end{equation}
There exists exactly three roots of the cubic equation (\ref{cubic-eq-1})
and each root defines uniquely $(y,z)$ and hence $(\lambda_1^2,\lambda_2^2)$.
In order to complete the proof of Theorem \ref{theorem-eig},
it remains to verify the exact representations (\ref{eig-real-intro}) and (\ref{eig-complex-intro})
for the eigenvalue pairs $\pm \lambda_1$, $\pm \lambda_2$, $\pm \lambda_3$.
The following two lemmas give the corresponding exact representations of eigenvalues.

\begin{lemma}
\label{lemma-eig-dn}
Assume $u_4 < u_3 < u_2 < u_1$ such that $u_1 + u_2 + u_3 + u_4 = 0$
and express $(c,d,e)$ by (\ref{u-roots}). The cubic equation (\ref{cubic-eq-1})
admits three simple roots in the explicit form:
\begin{eqnarray}
\label{root-real}
(R1) \;\; E_0 = \frac{1}{2} (u_1 u_4 + u_2 u_3), \quad
(R2) \;\; E_0 = \frac{1}{2} (u_1 u_3 + u_2 u_4), \quad
(R3) \;\; E_0 = \frac{1}{2} (u_1 u_2 + u_3 u_4).
\end{eqnarray}
\end{lemma}

\begin{proof}
Due to three possible choices of two eigenvalues $\lambda_1$ and $\lambda_2$ from three admissible values,
it is sufficient to consider one combination of the three eigenvalues,
e.g. the expression (\ref{eig-real-intro}) rewritten again in the form:
\begin{equation}
\label{eig-first}
\lambda_1 = \frac{1}{2} (u_1 + u_2), \quad \lambda_2 = \frac{1}{2} (u_1 + u_3), \quad
\lambda_3 = \frac{1}{2} (u_2 + u_3).
\end{equation}
Here $(\lambda_1,\lambda_2)$ are selected as eigenvalues of the algebraic method, whereas $\lambda_3$ is the
third complementary eigenvalue. We prove validity of the root (R1) in (\ref{root-real}) for $E_0$.
If $(\lambda_1,\lambda_2,\lambda_3)$ is given by (\ref{eig-first}),
then root (R1) for $E_0$ is equivalent to the following representation:
\begin{eqnarray}
\label{represent-E}
E_0 & = & \lambda_3^2 - \lambda_1^2 - \lambda_2^2 \\
\nonumber
& = & \frac{1}{2} (u_2 u_3 - u_1^2 - u_1 u_2 - u_1 u_3) \\
\nonumber
& = & \frac{1}{2} (u_2 u_3 + u_1 u_4).
\end{eqnarray}
Substituting (\ref{eig-first}) and (\ref{represent-E})
into (\ref{system-parameters}) yields the following relations:
\begin{eqnarray}
\label{represent-c}
c & = & 2 (\lambda_1^2 + \lambda_2^2 + \lambda_3^2) \\
\nonumber & = & (u_1 + u_2 + u_3)^2 - u_1 u_2 - u_1 u_3 - u_2 u_3 \\
\nonumber & = & - u_1 u_4 - u_2 u_4 - u_3 u_4 - u_1 u_2 - u_1 u_3 - u_2 u_3,
\end{eqnarray}
\begin{eqnarray}
\label{represent-d}
d & = & \lambda_1^4 + \lambda_2^4 + \lambda_3^4 - 2(\lambda_1^2 \lambda_2^2 + \lambda_1^2 \lambda_3^2 + \lambda_2^2 \lambda_3^2) \\
\nonumber & = & -u_1^2 u_2 u_3 - u_1 u_2^2 u_3 - u_1 u_2 u_3^2 \\
\nonumber & = & u_1 u_2 u_3 u_4,
\end{eqnarray}
and
\begin{eqnarray}
\label{represent-e}
e & = & - 4 \lambda_1 \lambda_2 \lambda_3 \\
\nonumber & = & \frac{1}{2} (u_1 u_2 u_3 - (u_1 + u_2 + u_3)(u_1 u_2 + u_1 u_3 + u_2 u_3)) \\
\nonumber & = & \frac{1}{2} (u_1 u_2 u_3 + u_1 u_2 u_4 + u_1 u_3 u_4 + u_2 u_3 u_4).
\end{eqnarray}
These relations confirm validity of the parametrization in (\ref{u-roots}). Hence,
the value (R1) in (\ref{root-real}) gives one of the three roots of
the cubic equation (\ref{cubic-eq-1}). The other
two roots in (R2) and (R3) in (\ref{root-real}) are given by the interchanges
$\lambda_2 \leftrightarrow \lambda_3$ and $\lambda_1 \leftrightarrow \lambda_3$ respectively
in the list of three eigenvalues in (\ref{eig-first}).
Since $u_4 < u_3 < u_2 < u_1$, the three roots (R1), (R2), and (R3) for $E_0$ in (\ref{root-real})
are distinct and no other roots of the cubic equation (\ref{cubic-eq-1}) exists.
\end{proof}

\begin{remark}
When $u_1 = u_2$ or $u_2 = u_3$ or $u_3 = u_4$ in Examples \ref{example-2}, \ref{example-3}, and \ref{example-4},
one eigenvalue in (\ref{eig-first}) is simple and the other eigenvalue is double.
This corresponds to one simple and one double roots for $E_0$ in (\ref{root-real}).
\end{remark}

\begin{remark}
When $e = 0$ in Example \ref{example-1}, one eigenvalue is zero, e.g., $\lambda_3 = 0$
in (\ref{eig-first}). Setting $u_1 = 1$ and $u_2 = \sqrt{1-k^2}$
for the normalized ${\rm dn}$-periodic wave (\ref{dn-intro}) yields the other two eigenvalues
$(\lambda_1,\lambda_2)$ in the form (\ref{eig-dn-intro}).
\end{remark}

\begin{lemma}
\label{lemma-eig-cn}
Assume $\beta \neq 0$, $\alpha = -(a+b)/2$ and express $(c,d,e)$ by (\ref{u-roots-complex}).
The cubic equation (\ref{cubic-eq-1}) admits three simple roots in the explicit form:
\begin{eqnarray}
\label{root-complex}
(R1) \;\; E_0 = \frac{1}{8} (a^2 + 6ab + b^2) + \frac{1}{2} \beta^2, \quad
(R2_{\pm}) \;\; E_0 = -\frac{1}{4} (a+b)^2 \pm \frac{i}{2} \beta (a-b).
\end{eqnarray}
\end{lemma}

\begin{proof}
Roots (\ref{root-complex}) are obtained from roots (\ref{root-real})
with the formal correspondence: $u_1 = a$, $u_2 = \alpha + i \beta$, $u_3 = \alpha - i \beta$,
$u_4 = b$. The constraint $u_1 + u_2 + u_3 + u_4 = 0$ is equivalent to $a + b + 2 \alpha = 0$,
which allows us to eliminate $\alpha = -(a+b)/2$ from all expressions and
to obtain eigenvalues in the form:
\begin{eqnarray}
\label{eig-first-complex}
\lambda_1 = \frac{1}{4} (a-b) + \frac{i}{2} \beta, \quad \lambda_2 = \frac{1}{4} (a-b) - \frac{i}{2} \beta,
\quad \lambda_3 = \frac{1}{2} (a+b).
\end{eqnarray}
The three roots for $E_0$ in (\ref{root-complex})
and the three eigenvalues $(\lambda_1,\lambda_2,\lambda_3)$
in (\ref{eig-first-complex}) are distinct if $\beta \neq 0$.
\end{proof}

\begin{remark}
When $e = 0$ in Example \ref{example-5}, one eigenvalue is zero, e.g., $\lambda_3 = 0$
in (\ref{eig-first-complex}). Setting $a = k$ and $\beta = \sqrt{1-k^2}$
for the normalized ${\rm cn}$-periodic wave (\ref{cn-intro}) yields the other two eigenvalues
$(\lambda_1,\lambda_2)$ in the form (\ref{eig-cn-intro}).
\end{remark}

\begin{remark}
The explicit expressions (\ref{eig-first}) of the three eigenvalues for the periodic wave (\ref{Jacob-1})
can be recovered from formulas (58) and (62) in \cite{K1}, where they are derived by a different technique.
The explicit expressions (\ref{eig-first-complex}) of the three eigenvalues for the periodic wave
(\ref{Jacob-2}) are new to the best of our knowledge.
\end{remark}

\section{Proof of Theorems \ref{theorem-DT-dn} and \ref{theorem-DT-cn}}
\label{sec-5}

Let us first establish algebraic relations on the squared eigenfunctions in the periodic spectral problem (\ref{3.2})
with the periodic solution $u$ to the mKdV equation (\ref{mKdV}).
We rewrite equations (\ref{potential}), (\ref{potential-DE}), (\ref{potential-DE2}), and (\ref{potential-DE3}) as a system
of linear equations for squared eigenfunctions:
\begin{eqnarray}
\label{alg-system-1}
\left\{ \begin{array}{rcl}
p_1^2 + q_1^2 + p_2^2 + q_2^2 & = & u, \\
2 \lambda_1 (p_1^2 - q_1^2) + 2 \lambda_2 (p_2^2 - q_2^2) & = & u_x, \\
4 \lambda_1^2 (p_1^2 + q_1^2) + 4 \lambda_2^2 (p_2^2 + q_2^2) & = & u_{xx} + 2 u^3 - 2 E_0 u, \\
8 \lambda_1^3 (p_1^2 - q_1^2) + 8 \lambda_2^3 (p_2^2 - q_2^2) & = & u_{xxx} + 6 u^2 u_x - 2 E_0 u_x,
\end{array} \right.
\end{eqnarray}
where the subscripts denote the derivatives in $x$.
Solving the algebraic system (\ref{alg-system-1}) with Cramer's rule yields the relations
\begin{eqnarray}
\label{relation-1}
p_1^2 + q_1^2 & = & \frac{u_{xx} + 2 u^3 - 2 E_0 u - 4 \lambda_2^2 u}{4 (\lambda_1^2 - \lambda_2^2)}, \\
\label{relation-2}
p_1^2 - q_1^2 & = & \frac{u_{xxx} + 6 u^2 u_x- 2 E_0 u_x - 4 \lambda_2^2 u_x}{8 \lambda_1 (\lambda_1^2 - \lambda_2^2)}
\end{eqnarray}
and similar relations for $p_2^2 + q_2^2$ and $p_2^2 - q_2^2$
obtained by the transformation $\lambda_1 \leftrightarrow \lambda_2$.
If $u$ satisfies differential equations (\ref{third-order}) and (\ref{second-order}), the expressions
(\ref{relation-1}) and (\ref{relation-2}) are simplified to the form
\begin{eqnarray}
\label{relation-1-simple}
p_1^2 + q_1^2 = \frac{e + 4 \lambda_1^2 u}{4(\lambda_1^2 - \lambda_2^2)}, \\
\label{relation-2-simple}
p_1^2 - q_1^2 = \frac{2 \lambda_1 u_x}{4 (\lambda_1^2 - \lambda_2^2)}.
\end{eqnarray}

We also rewrite equations (\ref{potential-squared}) and (\ref{potential-squared-1}) as another system
of linear equations for squared eigenfunctions:
\begin{eqnarray}
\label{alg-system-2}
\left\{ \begin{array}{rcl}
4 \lambda_1 p_1 q_1 + 4 \lambda_2 p_2 q_2 & = & E_0 - u^2, \\
16 \lambda_1^3 p_1 q_1 + 16 \lambda_2^3 p_2 q_2 & = & E_1 + E_0^2 + (u_x)^2 - 2 u u_{xx} - 3 u^4 + 2 E_0 u^2.
\end{array} \right.
\end{eqnarray}
Solving the algebraic system (\ref{alg-system-2}) with Cramer's rule yields the relations
\begin{eqnarray}
\label{relation-3}
4 \lambda_1 p_1 q_1 = \frac{E_1 + E_0^2 + (u_x)^2 - 2 u u_{xx} - 3 u^4 + 2 E_0 u^2 + 4 \lambda_2^2 u^2 - 4 \lambda_2^2 E_0}{4 (\lambda_1^2 - \lambda_2^2)}
\end{eqnarray}
and a similar relation for $4 \lambda_2 p_2 q_2$
obtained by the transformation $\lambda_1 \leftrightarrow \lambda_2$.
If $u$ satisfies the differential equation (\ref{first-order}), while $d$ and $E_1$ are expressed by
(\ref{parameter-de}) and (\ref{parameter-E-1}), then the expression (\ref{relation-3}) is simplified to the form
\begin{eqnarray}
\label{relation-3-simple}
p_1 q_1 = \frac{\lambda_1 (E_0 + 2 \lambda_2^2 - u^2)}{4 (\lambda_1^2 - \lambda_2^2)}.
\end{eqnarray}
By using the relations (\ref{represent-E}) and (\ref{represent-c}),
the relation (\ref{relation-3-simple}) can be rewritten in the equivalent form:
\begin{eqnarray}
\label{relation-3-simple-better}
2 p_1 q_1 = \frac{\lambda_1 (c - 4 \lambda_1^2 - 2 u^2)}{4 (\lambda_1^2 - \lambda_2^2)}.
\end{eqnarray}

\begin{remark}
Any eigenfunction $\varphi = (p_1,q_1)^t$ of the spectral problem (\ref{3.2})
is defined up to the scalar multiplication, and so are the Darboux transformations
(\ref{one-fold-general}), (\ref{two-fold-general}), and (\ref{three-fold-general}).
Therefore, the denominators in (\ref{relation-1-simple}), (\ref{relation-2-simple}), and (\ref{relation-3-simple-better})
can be canceled without loss of generality. On the other hand,
the numerators in (\ref{relation-1-simple}), (\ref{relation-2-simple}), and (\ref{relation-3-simple-better})
for the eigenfunction $\varphi = (p_1,q_1)^t$ can be extended to the other two eigenfunctions
$\varphi = (p_2,q_2)^t$ and $\varphi = (p_3,q_3)^t$ by replacing $\lambda_1 \mapsto \lambda_2$
and $\lambda_1 \mapsto \lambda_3$ respectively.
\end{remark}

We now proceed with the proof of Theorem \ref{theorem-DT-dn}. The following
three lemmas represent the outcomes of the one-fold, two-fold, and three-fold
Darboux transformations for the periodic wave (\ref{Jacob-1-intro})
with the periodic eigenfunctions of the spectral problem (\ref{3.2}).

\begin{lemma}
\label{lemma-DT-1}
Assume $u_4 < u_3 < u_2 < u_1$ such that $u_1 + u_2 + u_3 + u_4 = 0$,
$u_1 + u_2 \neq 0$, $u_1 + u_3 \neq 0$, and $u_2 + u_3 \neq 0$.
The one-fold Darboux transformation (\ref{one-fold-general})
with the periodic eigenfunction for each eigenvalue in (\ref{eig-real-intro})
transforms the periodic wave (\ref{Jacob-1-intro}) to the periodic wave of the same
period obtained after the corresponding symmetry transformation in
(\ref{symm-intro}) and the reflection $u \mapsto -u$.
\end{lemma}

\begin{proof}
Substituting (\ref{relation-1-simple}) and (\ref{relation-3-simple-better}) into (\ref{one-fold-general})
yields the new solution to the mKdV equation (\ref{mKdV}) in the form:
\begin{equation}
\label{one-fold}
\tilde{u} = \frac{e u + 2 \lambda_1^2 (c - 4 \lambda_1^2)}{e + 4 \lambda_1^2 u}.
\end{equation}
The representation (\ref{one-fold})
is a linear fractional transformation between two solutions to the mKdV equation (\ref{mKdV}).
By replacing $\lambda_1 \mapsto \lambda_2$ and $\lambda_1 \mapsto \lambda_3$ in (\ref{one-fold}),
two more solutions $\tilde{u}$ can be obtained from the representation (\ref{one-fold}).

We show with explicit computations that the transformation (\ref{one-fold}) applied to the periodic wave (\ref{Jacob-1})
with the eigenvalue $\lambda_j$ produces the same periodic wave after the corresponding symmetry transformation
(Sj) and the reflection $u \mapsto -u$.

By using (\ref{Jacob-1}), (\ref{u-roots}), and (\ref{eig-first}), we perform lengthy but straightforward computations to
obtain the following expressions:
\begin{equation}
\label{formula-1}
e + 4 \lambda_1^2 u = \frac{1}{2} (u_1 + u_2) (u_1-u_4) (u_2-u_4)
\frac{(u_1 - u_3) + (u_2-u_1)  {\rm sn}^2(\nu x;\kappa)}{(u_2 - u_4) + (u_1 - u_2)  {\rm sn}^2(\nu x;\kappa)}
\end{equation}
and
\begin{equation}
\label{formula-2}
e u + 2 \lambda_1^2 (c - 4 \lambda_1^2) =
\frac{1}{2} (u_1 + u_2) (u_1 - u_4) (u_2 - u_4)
\frac{(u_3 - u_1) u_2 + (u_1-u_2) u_3 {\rm sn}^2(\nu x;\kappa)}{(u_2 - u_4) + (u_1 - u_2)  {\rm sn}^2(\nu x;\kappa)},
\end{equation}
where $\nu$ and $\kappa$ are given by (\ref{roots-1}). Substituting (\ref{formula-1}) and (\ref{formula-2}) into (\ref{one-fold})
for $u_1 + u_2 \neq 0$ produces a new solution in the form:
\begin{equation}
\label{Jacob-1-new}
\tilde{u}(x) = - \left[ u_3 + \frac{(u_2 - u_3) (u_1 - u_3)}{(u_1 - u_3) + (u_2 - u_1) {\rm sn}^2(\nu x;\kappa)} \right].
\end{equation}
The new solution is obtained from the periodic wave (\ref{Jacob-1}) by the symmetry transformation (S1) in (\ref{symm-intro})
and the reflection $u \mapsto -u$, see the explicit form (\ref{Jacob-1-symm-a}).

By repeating the previous computations for the eigenvalue $\lambda_2$ in (\ref{eig-first}),
we obtain the following expressions:
\begin{equation}
\label{formula-3}
e + 4 \lambda_2^2 u = \frac{1}{2} (u_1 + u_3) (u_1-u_4) (u_1-u_2)
\frac{(u_2 - u_4) + (u_4-u_3)  {\rm sn}^2(\nu x;\kappa)}{(u_2 - u_4) + (u_1 - u_2)  {\rm sn}^2(\nu x;\kappa)}
\end{equation}
and
\begin{equation}
\label{formula-4}
e u + 2 \lambda_2^2 (c - 4 \lambda_2^2) =
\frac{1}{2} (u_1 + u_3) (u_1 - u_4) (u_1 - u_2)
\frac{(u_4 - u_2) u_3 + (u_3-u_4) u_2 {\rm sn}^2(\nu x;\kappa)}{(u_2 - u_4) + (u_1 - u_2)  {\rm sn}^2(\nu x;\kappa)}.
\end{equation}
Substituting (\ref{formula-3}) and (\ref{formula-4}) into the one-fold transformation
for $u_1 + u_3 \neq 0$ produces a new solution in the form:
\begin{equation}
\label{Jacob-1-new2}
\tilde{u}(x) = - \left[ u_2 + \frac{(u_4 - u_2) (u_3 - u_2)}{(u_4 - u_2) + (u_3 - u_4) {\rm sn}^2(\nu x;\kappa)} \right],
\end{equation}
which is obtained from the periodic wave (\ref{Jacob-1}) by the symmetry transformation (S2) in (\ref{symm-intro})
and the reflection $u \mapsto -u$, see the explicit form (\ref{Jacob-1-symm}).

By repeating the previous computations for the eigenvalue $\lambda_3$ in (\ref{eig-first}),
we obtain the following expressions:
\begin{equation}
\label{formula-5}
e + 4 \lambda_3^2 u = \frac{1}{2} (u_2 + u_3) (u_1-u_2) (u_2-u_4)
\frac{(u_3 - u_1) + (u_4-u_3)  {\rm sn}^2(\nu x;\kappa)}{(u_2 - u_4) + (u_1 - u_2)  {\rm sn}^2(\nu x;\kappa)}
\end{equation}
and
\begin{equation}
\label{formula-6}
e u + 2 \lambda_3^2 (c - 4 \lambda_3^2) =
\frac{1}{2} (u_2 + u_3) (u_1 - u_2) (u_2 - u_4)
\frac{(u_1 - u_3) u_4 + (u_3-u_4) u_1 {\rm sn}^2(\nu x;\kappa)}{(u_2 - u_4) + (u_1 - u_2)  {\rm sn}^2(\nu x;\kappa)},
\end{equation}
Substituting (\ref{formula-5}) and (\ref{formula-6}) into the one-fold transformation
for $u_2 + u_3 \neq 0$ produces a new solution in the form:
\begin{equation}
\label{Jacob-1-new3}
\tilde{u}(x) = - \left[ u_1 + \frac{(u_3 - u_1) (u_4 - u_1)}{(u_3 - u_1) + (u_4 - u_3) {\rm sn}^2(\nu x;\kappa)} \right],
\end{equation}
which is obtained from the periodic wave (\ref{Jacob-1}) by the symmetry transformation (S3) in (\ref{symm-intro})
and the reflection $u \mapsto -u$, see the explicit form (\ref{Jacob-1-symm-b}).
\end{proof}

\begin{remark}
If $e = 0$, the one-fold transformation (\ref{one-fold}) simplifies to the form
\begin{equation}
\label{one-fold-e-0}
\tilde{u} = \frac{c - 4 \lambda_1^2}{2 u}.
\end{equation}
Setting $u_1 = 1$ and $u_2 = \sqrt{1-k^2}$ yields the normalized dnoidal periodic wave (\ref{dn-intro})
with the two nonzero eigenvalues (\ref{eig-dn-intro}). Substituting these expressions into (\ref{one-fold-e-0}) yields
\begin{equation}
\label{dn-transformed}
\tilde{u}(x) = \mp\frac{\sqrt{1-k^2}}{{\rm dn}(x;k)} = \mp{\rm dn}(x+K(k);k) = \mp u(x+L/2),
\end{equation}
where $L = 2K(k)$ is the period of the periodic wave (\ref{dn-intro}).
Hence, the transformed solution (\ref{dn-transformed}) is again a translated and reflected
version of the periodic wave (\ref{dn-intro}).
\end{remark}

\begin{lemma}
\label{lemma-DT-2}
Under the same assumptions as in Lemma \ref{lemma-DT-1},
the two-fold Darboux transformation (\ref{two-fold-general})
with the periodic eigenfunctions for any two eigenvalues from (\ref{eig-real-intro})
transforms the periodic wave (\ref{Jacob-1-intro}) to the periodic wave of the same
period obtained after the complementary third symmetry transformation in
(\ref{symm-intro}).
\end{lemma}

\begin{proof}
Substituting (\ref{relation-1-simple}), (\ref{relation-2-simple}), and (\ref{relation-3-simple})
into (\ref{two-fold-general}) yields the new solution to the mKdV equation (\ref{mKdV}) in the form:
\begin{eqnarray*}
\tilde{u} = u + \frac{4(\lambda_1^2 - \lambda_2^2)^2 [8 \lambda_1^2 \lambda_2^2 u - e( E_0 - u^2)]}{
8 \lambda_1^2 \lambda_2^2 [(u_x)^2 + (E_0+2\lambda_1^2 - u^2)(E_0+2\lambda_2^2-u^2)]
- (\lambda_1^2+\lambda_2^2) (e + 4 \lambda_1^2 u ) (e + 4 \lambda_2^2 u)}.
\end{eqnarray*}
By using the expressions (\ref{zero-order}), (\ref{parameter-de}), and (\ref{parameter-e}),
we obtain
\begin{eqnarray}
\tilde{u} = \frac{e E_0 - 4 \lambda_1^2 \lambda_2^2 u}{eu + 4 \lambda_1^2 \lambda_2^2},
\label{two-fold}
\end{eqnarray}
where $E_0 = \lambda_3^2 - \lambda_1^2 - \lambda_2^3$ with $\lambda_3$ being the complementary
third eigenvalue to the pair $(\lambda_1,\lambda_2)$. Any pair of two eigenvalues
from the list of three eigenvalues in (\ref{eig-first}) can be picked as $(\lambda_1,\lambda_2)$.

The representation (\ref{two-fold}) is another linear
fractional transformation between two solutions of the mKdV equation (\ref{mKdV}). We show
with explicit computations that this transformation applied to the periodic wave (\ref{Jacob-1})
with the eigenvalue pair $(\lambda_1,\lambda_2)$ produces the same periodic wave after
the complementary symmetry transformation (S3).

By using  (\ref{Jacob-1}), (\ref{u-roots}), and (\ref{eig-first}), we obtain the following expressions:
$$
eu + 4 \lambda_1^2 \lambda_2^2 = \frac{1}{4} (u_1 - u_2) (u_2 - u_4) (u_1 + u_2) (u_1 + u_3)
\frac{(u_1 - u_3) + (u_3-u_4)  {\rm sn}^2(\nu x;\kappa)}{(u_2-u_4) + (u_1 - u_2)  {\rm sn}^2(\nu x;\kappa)}
$$
and
$$
e E_0 - 4 \lambda_1^2 \lambda_2^2 u =
\frac{1}{4} (u_1 - u_2) (u_2 - u_4) (u_1 + u_2) (u_1 + u_3)
\frac{u_4 (u_1 - u_3) + u_1 (u_3 - u_4) {\rm sn}^2(\nu x;\kappa)}{(u_2-u_4) + (u_1 - u_2)  {\rm sn}^2(\nu x;\kappa)},
$$
where $\nu$ and $\kappa$ are given by (\ref{roots-1}). Substituting these expressions into (\ref{two-fold})
for $u_1 + u_2 \neq 0$ and $u_1 + u_3 \neq 0$ produces a new solution in the form:
\begin{equation}
\label{Jacob-1-new-two2}
\tilde{u}(x) = u_1 + \frac{(u_4 - u_1) (u_1 - u_3)}{(u_1 - u_3) + (u_3 - u_4) {\rm sn}^2(\nu x;\kappa)},
\end{equation}
The new solution is obtained from the periodic wave (\ref{Jacob-1}) by the symmetry transformation (S3) in (\ref{symm-intro}).
Repeating the same computations for the eigenvalue pair $(\lambda_1,\lambda_3)$ in (\ref{eig-first})
for $u_1 + u_2 \neq 0$ and $u_2 + u_3 \neq 0$ produces
the periodic wave (\ref{Jacob-1}) after the symmetry transformation (S2) in (\ref{symm-intro}).
Repeating the same computations for the eigenvalue pair $(\lambda_2,\lambda_3)$ in (\ref{eig-first})
for $u_1 + u_3 \neq 0$ and $u_2 + u_3 \neq 0$
produces the periodic wave (\ref{Jacob-1}) after the symmetry transformation (S1) in (\ref{symm-intro}).
\end{proof}

\begin{remark}
The two-fold transformation with two eigenvalues $(\lambda_1,\lambda_2)$ can be thought as a composition
of two one-fold transformations with eigenvalues $\lambda_1$ and $\lambda_2$. Indeed, by Lemma \ref{lemma-DT-1},
the one-fold transformation with $\lambda_1$ performs symmetry transformation (S1) and reflection,
whereas the one-fold transformation with $\lambda_2$ performs symmetry transformation (S2) and reflection.
Composition of (S1) and (S2) yields (S3), whereas two reflections annihilate each other.
\end{remark}

\begin{remark}
If $e = 0$, the two-fold transformation (\ref{two-fold}) yields $\tilde{u} = -u$.
Indeed, the one-fold Darboux transformations of the ${\rm dn}$-periodic wave (\ref{dn-intro})
with the two nonzero eigenvalues in (\ref{eig-dn-intro})
produce (\ref{dn-transformed}), a composition of which yields just the reflection of $u$.
\end{remark}

\begin{lemma}
\label{lemma-DT-3}
Under the same assumptions as in Lemma \ref{lemma-DT-1},
the three-fold Darboux transformation (\ref{three-fold-general})
with the periodic eigenfunctions for all three eigenvalues in (\ref{eig-real-intro})
transforms the periodic wave (\ref{Jacob-1-intro}) to itself reflected with $u \mapsto -u$.
\end{lemma}

\begin{proof}
By using (\ref{relation-3-simple-better}) with $c = 2(\lambda_1^2 + \lambda_2^2 + \lambda_3^2)$
and neglecting the denominators, we write the product terms in the form:
\begin{equation}\label{14.4} \left\{\begin{array}{lll}
\lambda_1 p_1 q_1 & \cong & \lambda^2_1 (\lambda_2^2+\lambda_3^2-\lambda_1^2-u^2),\\
\lambda_2 p_2 q_2 & \cong & \lambda^2_2 (\lambda_1^2+\lambda_3^2-\lambda_2^2-u^2),\\
\lambda_3 p_3 q_3 & \cong & \lambda^2_3 (\lambda_1^2+\lambda_2^2-\lambda_3^2-u^2),\\
\end{array}
\right.
\end{equation}
the sign $\cong$ denotes the missing multiplication factor.
By using (\ref{relation-1-simple}) and (\ref{relation-2-simple})
and neglecting the same denominators, we also express the squared components in the form:
\begin{equation}\label{14.5} \left\{\begin{array}{lll}
2 p_1^2 \cong e+4\lambda_1^2u+2\lambda_1u_x,\qquad
2 q_1^2 \cong e+4\lambda_1^2u-2\lambda_1u_x,\\
2 p_2^2 \cong e+4\lambda_2^2u+2\lambda_2u_x,\qquad
2 q_2^2 \cong e+4\lambda_2^2u-2\lambda_2u_x,\\
2 p_3^2 \cong e+4\lambda_3^2u+2\lambda_3u_x,\qquad
2 q_3^2 \cong e+4\lambda_3^2u-2\lambda_3u_x.
\end{array}
\right.
\end{equation}
Only for succinctness in writing, let us make the notation
$$
\Upsilon_1=e+4\lambda_1^2 u,\quad \Upsilon_2=e+4\lambda_2^2 u,\quad \Upsilon_3=e+4\lambda_3^2u.
$$
By substituting (\ref{14.4}), and (\ref{14.5})
back into the three-fold transformation formula (\ref{three-fold-general})
we simplify the numerator $N$ and the denominator $D$
with the straightforward computations to the form:
\begin{equation*}
\begin{array}{lll}
N&=& \lambda_3^2(\lambda_3^2-\lambda_1^2)
(\lambda_3^2-\lambda_2^2)(\lambda_1^2+\lambda_2^2-\lambda_3^2-u^2)[8\lambda_1^2\lambda_2^2(u^4+d) \\
&&+e^2(\lambda_1^2+\lambda_2^2-u^2)+4 e u(\lambda_1^2-\lambda_2^2)^2]+\lambda_2^2(\lambda_2^2-\lambda_1^2)
(\lambda_2^2-\lambda_3^2)(\lambda_1^2+\lambda_3^2-\lambda_2^2-u^2)\\
&&\times[8\lambda_1^2\lambda_3^2(u^4+d)
+e^2(\lambda_1^2+\lambda_3^2-u^2)+4 e u(\lambda_1^2-\lambda_3^2)^2]+\lambda_1^2(\lambda_1^2-\lambda_2^2)
(\lambda_1^2-\lambda_3^2)\\&&\times(\lambda_2^2+\lambda_3^2-\lambda_1^2-u^2)
[8\lambda_2^2\lambda_3^2(u^4+d)
+e^2(\lambda_2^2+\lambda_3^2-u^2)+4 e u(\lambda_2^2-\lambda_3^2)^2]-8\lambda_1^2\lambda_2^2\lambda_3^2\\
&& \times(\lambda_1^2+\lambda_2^2-\lambda_3^2-u^2)(\lambda_1^2+\lambda_3^2-\lambda_2^2-u^2)
(\lambda_2^2+\lambda_3^2-\lambda_1^2-u^2)(d+\lambda_1^2\lambda_2^2+\lambda_1^2\lambda_3^2+\lambda_2^2\lambda_3^2),
\end{array}
\end{equation*}
and
\begin{equation*}
\begin{array}{lll}
D&=& (\lambda_1+\lambda_2)^2(\lambda_1+\lambda_3)^2(\lambda_2-\lambda_3)^2[\frac14 \Upsilon_1\Upsilon_2\Upsilon_3+u_x^2
(\lambda_2\lambda_3\Upsilon_1-\lambda_1\lambda_2\Upsilon_3 \\
&&-\lambda_1\lambda_3\Upsilon_2)]
+(\lambda_1+\lambda_2)^2(\lambda_2+\lambda_3)^2(\lambda_1-\lambda_3)^2[\frac14 \Upsilon_1\Upsilon_2\Upsilon_3+u_x^2
(\lambda_1\lambda_3\Upsilon_2-\lambda_1\lambda_2\Upsilon_3\\&&-\lambda_2\lambda_3\Upsilon_1)]
+(\lambda_1+\lambda_3)^2(\lambda_2+\lambda_3)^2(\lambda_1-\lambda_2)^2[\frac14 \Upsilon_1\Upsilon_2\Upsilon_3+u_x^2
(\lambda_1\lambda_2\Upsilon_3-\lambda_1\lambda_3\Upsilon_2\\&&-\lambda_2\lambda_3\Upsilon_1)]
+(\lambda_1-\lambda_2)^2(\lambda_2-\lambda_3)^2(\lambda_1-\lambda_3)^2[\frac14 \Upsilon_1\Upsilon_2\Upsilon_3+u_x^2
(\lambda_1\lambda_3\Upsilon_2+\lambda_1\lambda_2\Upsilon_3\\&&+\lambda_2\lambda_3\Upsilon_1)]-
8\lambda_1^2\lambda_2^2(\lambda_3^2-\lambda_1^2)(\lambda_3^2-\lambda_2^2)\Upsilon_3(\lambda_2^2+\lambda_3^2-\lambda_1^2-u^2)
(\lambda_1^2+\lambda_3^2-\lambda_2^2-u^2)\\
&&-
8\lambda_1^2\lambda_3^2(\lambda_2^2-\lambda_1^2)(\lambda_2^2-\lambda_3^2)\Upsilon_2(\lambda_2^2+\lambda_3^2-\lambda_1^2-u^2)
(\lambda_1^2+\lambda_2^2-\lambda_3^2-u^2)\\
&&-
8\lambda_2^2\lambda_3^2(\lambda_1^2-\lambda_2^2)(\lambda_1^2-\lambda_3^2)\Upsilon_1(\lambda_1^2+\lambda_3^2-\lambda_2^2-u^2)
(\lambda_1^2+\lambda_2^2-\lambda_3^2-u^2).
\end{array}
\end{equation*}
By using (\ref{zero-order}) to express $(u_x)^2$ and
using (\ref{represent-c}), (\ref{represent-d}), and (\ref{represent-e})
to express $(c,d,e)$, we obtain
$$
N = 32 \lambda_1 \lambda_2 \lambda_3 (\lambda_1^2 - \lambda_2^2)^2 (\lambda_1^2-\lambda_3^2)^2 (\lambda_2^2 - \lambda_3^2)^2 u = -\frac{1}{2} u D,
$$
which yields $\tilde{u}=u-2u=-u$.
\end{proof}

We end this section with the proof of Theorem \ref{theorem-DT-cn}. The following
three lemmas represent the outcomes of the one-fold, two-fold, and three-fold
Darboux transformations for the periodic wave (\ref{Jacob-2-intro}) with the periodic
eigenfunctions of the spectral problem (\ref{3.2}). Only real-valued solutions $\tilde{u}$
to the mKdV equation (\ref{mKdV}) are allowed as the outcomes of the Darboux transformations.

\begin{lemma}
\label{lemma-DT-4}
Assume $a \neq \pm b$. The one-fold Darboux transformation (\ref{one-fold-general})
with the periodic eigenfunction for eigenvalue $\lambda_3$ in (\ref{eig-complex-intro})
transforms the periodic wave (\ref{Jacob-2-intro}) to the periodic wave of the same
period obtained after the symmetry transformation (S0) in (\ref{symm-4-intro}) and the reflection $u \mapsto -u$.
\end{lemma}

\begin{proof}
By using (\ref{Jacob-2}), (\ref{u-roots-complex}), and (\ref{eig-first-complex}),
we obtain the following expressions:
$$
e + 4 \lambda_3^2 u = \frac{1}{2} (a+b) \nu^2 \frac{(1 + \delta) + (1-\delta)  {\rm cn}(\nu x;\kappa)}{(1 + \delta) + (\delta-1)  {\rm cn}(\nu x;\kappa)}
$$
and
$$
e u + 2 \lambda_3^2 (c - 4 \lambda_3^2) = -\frac{1}{2} (a+b) \nu^2
\frac{(b + a \delta) + (b - a \delta)  {\rm cn}(\nu x;\kappa)}{(1 + \delta) + (\delta-1)  {\rm cn}(\nu x;\kappa)},
$$
where $\delta$, $\nu$, and $\kappa$ are defined by (\ref{roots-2}) with $\alpha = -(a+b)/2$.
Substituting these expressions into (\ref{one-fold}) with $\lambda_1 \mapsto \lambda_3$ for $a + b \neq 0$ produces a new solution
\begin{equation}
\label{Jacob-2-new}
\tilde{u}(x) = - \left[ b + \frac{(a-b) (1- {\rm cn}(\nu x;\kappa))}{1 + \delta^{-1} + (\delta^{-1} - 1)  {\rm cn}(\nu x;\kappa)} \right],
\end{equation}
which is obtained from the periodic wave (\ref{Jacob-2}) by the symmetry transformation (S0) in (\ref{symm-4-intro})
and the reflection $u \mapsto -u$, see the explicit form (\ref{Jacob-2-symm}).
\end{proof}

\begin{remark}
The one-fold transformation (\ref{one-fold}) with $\lambda_1 \mapsto \lambda_3$ cannot be used for $e = 0$,
because $\lambda_3 = 0$ and $(p_3,q_3) = (0,0)$ if $e = 0$.
\end{remark}

\begin{lemma}
\label{lemma-DT-5}
Under the same assumptions as in Lemma \ref{lemma-DT-4},
the two-fold Darboux transformation (\ref{two-fold-general})
with the periodic eigenfunctions for two eigenvalues $\lambda_1$ and $\lambda_2$ from (\ref{eig-complex-intro})
transforms the periodic wave (\ref{Jacob-2-intro}) to the periodic wave of the same
period obtained after the symmetry transformation (S0) in
(\ref{symm-4-intro}).
\end{lemma}

\begin{proof}
By using (\ref{Jacob-2}), (\ref{u-roots-complex}),
and (\ref{eig-first-complex}), we obtain the following expressions:
$$
eu + 4 \lambda_1^2 \lambda_2^2 = \frac{1}{4} \left[ \frac{1}{4} (a-b)^2 + \beta^2 \right] \nu^2
\frac{(1 + \delta) + (1-\delta)  {\rm cn}(\nu x;\kappa)}{(1 + \delta) + (\delta-1)  {\rm cn}(\nu x;\kappa)}
$$
and
$$
e E_0 - 4 \lambda_1^2 \lambda_2^2 u = \frac{1}{4} \left[ \frac{1}{4} (a-b)^2 + \beta^2 \right]  \nu^2
\frac{(b + a \delta) + (b - a \delta)  {\rm cn}(\nu x;\kappa)}{(1 + \delta) + (\delta-1)  {\rm cn}(\nu x;\kappa)},
$$
where $\delta$, $\nu$, and $\kappa$ are defined by (\ref{roots-2}) with $\alpha = -(a+b)/2$.
Substituting these expressions into (\ref{two-fold}) for $a + b \neq 0$ produces a new solution
\begin{equation}
\label{Jacob-2-new-two}
\tilde{u}(x) = b + \frac{(a-b) (1- {\rm cn}(\nu x;\kappa))}{1 + \delta^{-1} + (\delta^{-1} - 1)  {\rm cn}(\nu x;\kappa)},
\end{equation}
which is obtained from the periodic wave (\ref{Jacob-2}) by the symmetry transformation (S0) in (\ref{symm-4-intro}).
\end{proof}

\begin{remark}
If $e = 0$, then the two-fold transformation (\ref{two-fold}) produces $\tilde{u} = -u$, which is just a reflection of $u$.
\end{remark}

\begin{lemma}
\label{lemma-DT-6}
Under the same assumptions as in Lemma \ref{lemma-DT-4},
the three-fold Darboux transformation (\ref{three-fold-general})
with the periodic eigenfunctions for all three eigenvalues in (\ref{eig-complex-intro})
transforms the periodic wave (\ref{Jacob-2-intro}) to itself reflected with $u \mapsto -u$.
\end{lemma}

\begin{proof}
Computations in the proof of Lemma \ref{lemma-DT-3} are independent on the choice of $(\lambda_1,\lambda_2,\lambda_3)$,
hence they extend to the eigenvalues in (\ref{eig-complex-intro}).
\end{proof}

\section{Second solution of the Lax system}
\label{sec-6}

By Theorems \ref{theorem-DT-dn} and \ref{theorem-DT-cn}, Darboux transformations with periodic eigenfunctions of the spectral
problem (\ref{3.2}) only generate symmetry transformations of the periodic solutions to the mKdV equation (\ref{mKdV}).
In order to obtain new solutions to the mKdV equation (\ref{mKdV}), we construct the second linearly independent
solution to the spectral problem (\ref{3.2}) with the same eigenvalue $\lambda$. We also include the time evolution (\ref{3.3}) in all
expressions.

The following lemma gives the time evolution of the periodic eigenfunction
$\varphi$ satisfying the Lax equations (\ref{3.2})--(\ref{3.3})
if $u$ is the periodic travelling solution to the mKdV equation (\ref{mKdV}).

\begin{lemma}
\label{lemma-time-per}
Let $u(x,t) = u(x-ct)$ be a periodic travelling wave of the mKdV equation (\ref{mKdV}) with the wave speed $c$,
hence $u$ satisfies the third-order differential equation (\ref{third-order}).
Let $\varphi = (p_1,q_1)^t$ be the periodic eigenfunction of the spectral problem
(\ref{3.2}) with $\lambda = \lambda_1$.
Then, $\varphi(x,t) = \varphi(x-ct)$ satisfies the time evolution system (\ref{3.3}).
\end{lemma}

\begin{proof}
Since the relation (\ref{potential}) holds for every $t$ and $u(x,t) = u(x-ct)$, then $\varphi(x,t) = \varphi(x-ct)$.
Alternatively, by using (\ref{second-order}), (\ref{relation-1-simple}), (\ref{relation-2-simple}),
and (\ref{relation-3-simple}) in the time evolution problem (\ref{3.3}), we obtain
$$
\frac{\partial p_1}{\partial t} + c \frac{\partial p_1}{\partial x} = 0, \quad
\frac{\partial q_1}{\partial t} + c \frac{\partial q_1}{\partial x} = 0,
$$
hence $p_1(x,t) = p_1(x-ct)$ and $q_1(x,t) = q_1(x-ct)$.
\end{proof}

Let $\varphi = (p_1,q_1)^t$ be the periodic eigenfunction of the spectral problem (\ref{3.2}) with $\lambda = \lambda_1$
and denote the second linearly independent solution the same spectral problem (\ref{3.2})
with the same $\lambda = \lambda_1$ by $\varphi = (\hat{p}_1,\hat{q}_1)^t$. Since the coefficient
matrix in (\ref{3.2}) has zero trace, the Wronskian determinant between the two solutions
is constant in $x$ and nonzero. To keep consistency with our previous work \cite{CPkdv},
we normalize the Wronskian by $2$, hence
\begin{equation}
\label{normalization}
p_1 \hat{q}_1 - \hat{p}_1 q_1 = 2.
\end{equation}
In the previous work \cite{CPkdv}, we constructed the second solution in the explicit form:
\begin{equation}
\label{represent-old}
\hat{p}_1 = \frac{\theta_1 - 1}{q_1}, \quad \hat{q}_1 = \frac{\theta_1 + 1}{p_1},
\end{equation}
where $\theta_1$ satisfies a certain scalar equation in $x$ and $t$, which can be easily integrated.
The corresponding expressions were used in \cite{CPkdv} to construct the rogue waves on the periodic
background, however, it was found that the expressions may be undefined if there exists a point of $(x,t)$ for which
either $p_1$ or $q_1$ vanishes.

Here we consider a different representation of the second solution which is free
of the technical problem above. The following lemma represents the second
solution in the explicit form.

\begin{lemma}
\label{lemma-time-second}
Let $\varphi = (p_1,q_1)^t$ be the periodic eigenfunction satisfying the Lax equations (\ref{3.2})--(\ref{3.3})
with $\lambda = \lambda_1$ and $u(x,t) = u(x-ct)$ and let $\varphi$ satisfy the normalization conditions
(\ref{relation-1-simple}), (\ref{relation-2-simple}), and (\ref{relation-3-simple-better}).
The second linearly independent solution satisfying the normalization (\ref{normalization}) can be written in the form:
\begin{equation}
\label{represent-new}
\hat{p}_1 = p_1 \phi_1 - \frac{2 q_1}{p_1^2 + q_1^2}, \quad
\hat{q}_1 = q_1 \phi_1 + \frac{2 p_1}{p_1^2 + q_1^2},
\end{equation}
with
\begin{equation}
\label{phi-final}
\phi_1(x,t) = - 16 (\lambda_1^2 - \lambda_2^2) \left[ \lambda_1^2 \int_0^{x-ct} \frac{c - 4 \lambda_1^2 - 2 u^2}{(e + 4 \lambda_1^2 u)^2} dy
+ t + \psi_1 \right],
\end{equation}
where $\psi_1$ is independent of $(x,t)$.
\end{lemma}

\begin{proof}
By substituting the representation (\ref{represent-new}) to the spectral problem (\ref{3.2}) with $\lambda = \lambda_1$ and
using the same spectral problem (\ref{3.2}) for $\varphi = (p_1,q_1)^t$, we obtain
\begin{equation}
\label{phi-der-x}
\frac{d \phi_1}{dx} = -\frac{8 \lambda_1 p_1 q_1}{(p_1^2+q_1^2)^2}.
\end{equation}
Thanks to (\ref{relation-1-simple}) and (\ref{relation-3-simple-better}), this equation can be integrated to the form
\begin{equation}
\label{phi}
\phi_1 = -16 (\lambda_1^2 - \lambda_2^2) \left[ \lambda_1^2  \int_0^x \frac{c - 4 \lambda_1^2 - 2u^2}{(e + 4 \lambda_1^2 u)^2} dy + \psi_1 \right],
\end{equation}
where $\psi_1$ is the constant of integration in $x$ that may depend on $t$.

On the other hand, by substituting the representation (\ref{represent-new}) to the time evolution system (\ref{3.3})
with $\lambda_1$ and using the same system (\ref{3.3}) for $\varphi = (p_1,q_1)^t$, we obtain
\begin{equation}
\label{phi-der-t}
\frac{\partial \phi_1}{\partial t} = \frac{8 \lambda_1 p_1 q_1 (4 \lambda_1^2 + 2 u^2)}{(p_1^2+q_1^2)^2}
- \frac{8 \lambda_1 u_x (p_1^2 - q_1^2)}{(p_1^2+q_1^2)^2}.
\end{equation}
Substituting (\ref{relation-2-simple}), (\ref{relation-3-simple-better}), and (\ref{phi-der-x}) into (\ref{phi-der-t})
yields the scalar equation
$$
\frac{\partial \phi_1}{\partial t} + c \frac{\partial \phi_1}{\partial x} = -16 (\lambda_1^2 - \lambda_2^2),
$$
from which we obtain (\ref{phi}), where $\psi_1$ is now constant both in $x$ and $t$.
\end{proof}

\begin{remark}
Without loss of generality, thanks to the translational invariance in $(x,t)$, we can set $\psi_1 = 0$
so that if $u$ is even in $x$, then $\phi_1$ is odd in $x$ at $t = 0$.
\end{remark}

\begin{remark}
The representation (\ref{represent-new}) is non-singular for every $(x,t)$ for which $p_1^2 + q_1^2 \neq 0$.
If $\lambda_1 \in \mathbb{R}$ with real $\varphi = (p_1,q_1)^t$, we have $p_1^2 + q_1^2 \neq 0$ everywhere because
if $\varphi$ vanishes at one point, then $\varphi$ is identically zero everywhere since it satisfies
the first-order systems (\ref{3.2}) and (\ref{3.3}).
\end{remark}

\begin{remark}
If $\lambda_1 \notin \mathbb{R}$, the representation (\ref{represent-new}) is non-singular if $e \neq 0$
and singular if $e = 0$. This follows from the expression (\ref{relation-1-simple}) which shows that
$p_1^2 + q_1^2 \cong e + 4 \lambda_1^2 u$ with $e \in \mathbb{R}$ and $u \in \mathbb{R}$.
If $e = 0$, the {\em cn}-periodic solution (\ref{cn-intro}) vanish at some points of $(x,t)$, which
lead to the singular behavior of $\hat{p}_1$ and $\hat{q}_1$ given by (\ref{represent-new}).
Note in this case, the previous representation (\ref{represent-old}) is non-singular, as it follows
from our previous work \cite{CPkdv}.
\end{remark}

\section{Proof of Theorems \ref{theorem-rogue-dn} and \ref{theorem-rogue-cn}}
\label{sec-7}

Here we use the Darboux transformation formulas (\ref{one-fold-general}), (\ref{two-fold-general}),
and (\ref{three-fold-general}) with the second solutions of the Lax system (\ref{3.2}) and
(\ref{3.3}) for the same eigenvalues $\lambda_1$, $\lambda_2$, and $\lambda_3$ as in Theorem \ref{theorem-eig}.
Outcomes of the Darboux transformations are first represented graphically and then studied analytically.

Substituting (\ref{represent-new}) into the one-fold transformation (\ref{one-fold-general})
yields the new solution to the mKdV equation (\ref{mKdV}) in the form:
\begin{equation}
\label{one-fold-new}
\hat{u} = u + \frac{4 \lambda_1 N_1}{D_1},
\end{equation}
where
\begin{eqnarray*}
N_1 & := & \frac{p_1 q_1}{p_1^2 + q_1^2} \left[ (p_1^2 + q_1^2)^2 \phi_1^2 - 4 \right] + 2 (p_1^2 - q_1^2) \phi_1, \\
D_1 & := & (p_1^2 + q_1^2)^2 \phi_1^2 + 4,
\end{eqnarray*}
with $\phi_1$ given by (\ref{phi-final}) for the choice of $\psi_1 = 0$.

Figure \ref{fig-1-fold-1} shows three new solutions computed from
the periodic wave (\ref{Jacob-1-intro}) with the parameter values:
\begin{equation}
\label{choice-1}
u_1 = 2, \quad u_2 = -0.25, \quad u_3 = -0.75, \quad u_4 = -1.
\end{equation}
The three solutions are generated with three different choices
for the eigenvalue $\lambda_1$ given by (\ref{eig-real-intro}).
Each solution displays an algebraic soliton propagating on the background
of the periodic travelling wave.

\begin{figure}[ht]
\centering
\includegraphics[scale=0.4]{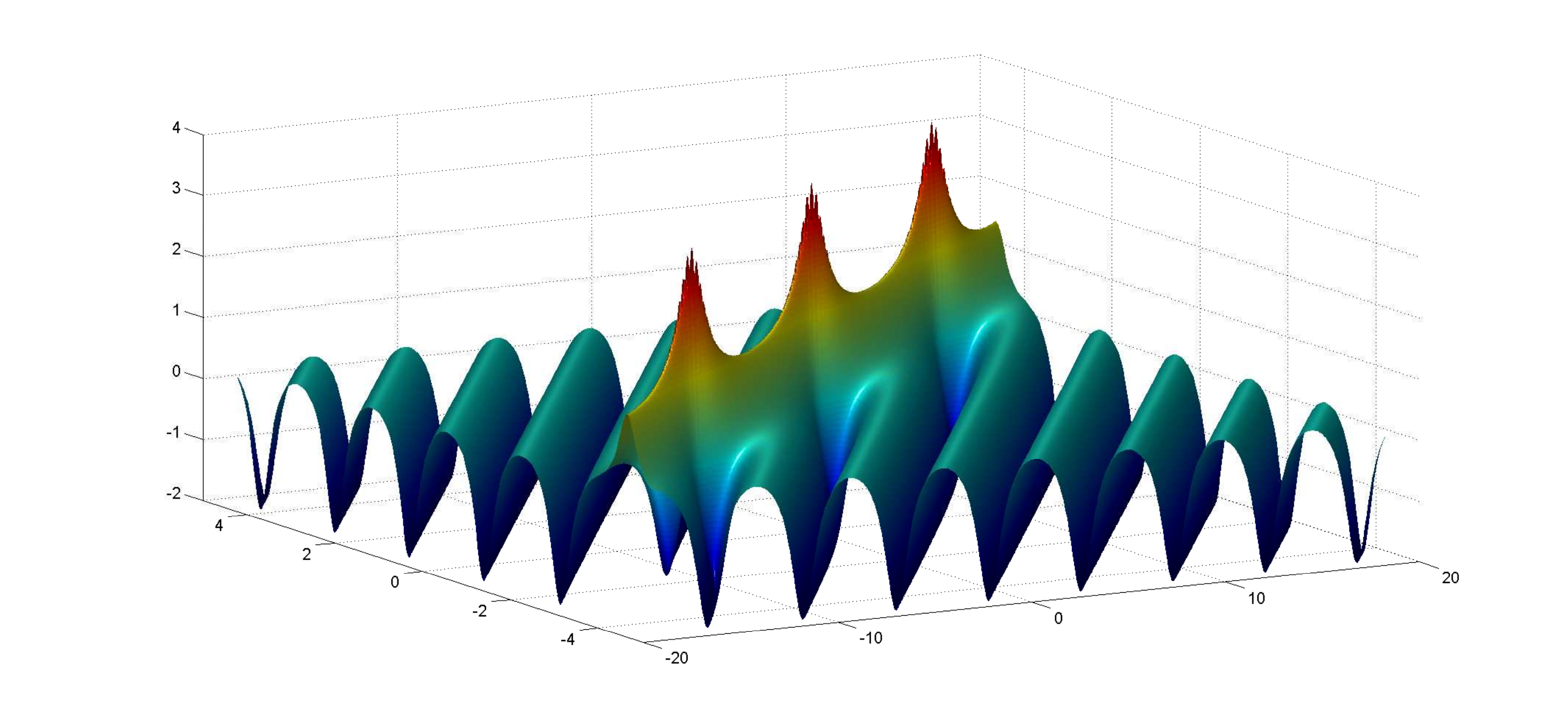}\\
\includegraphics[scale=0.4]{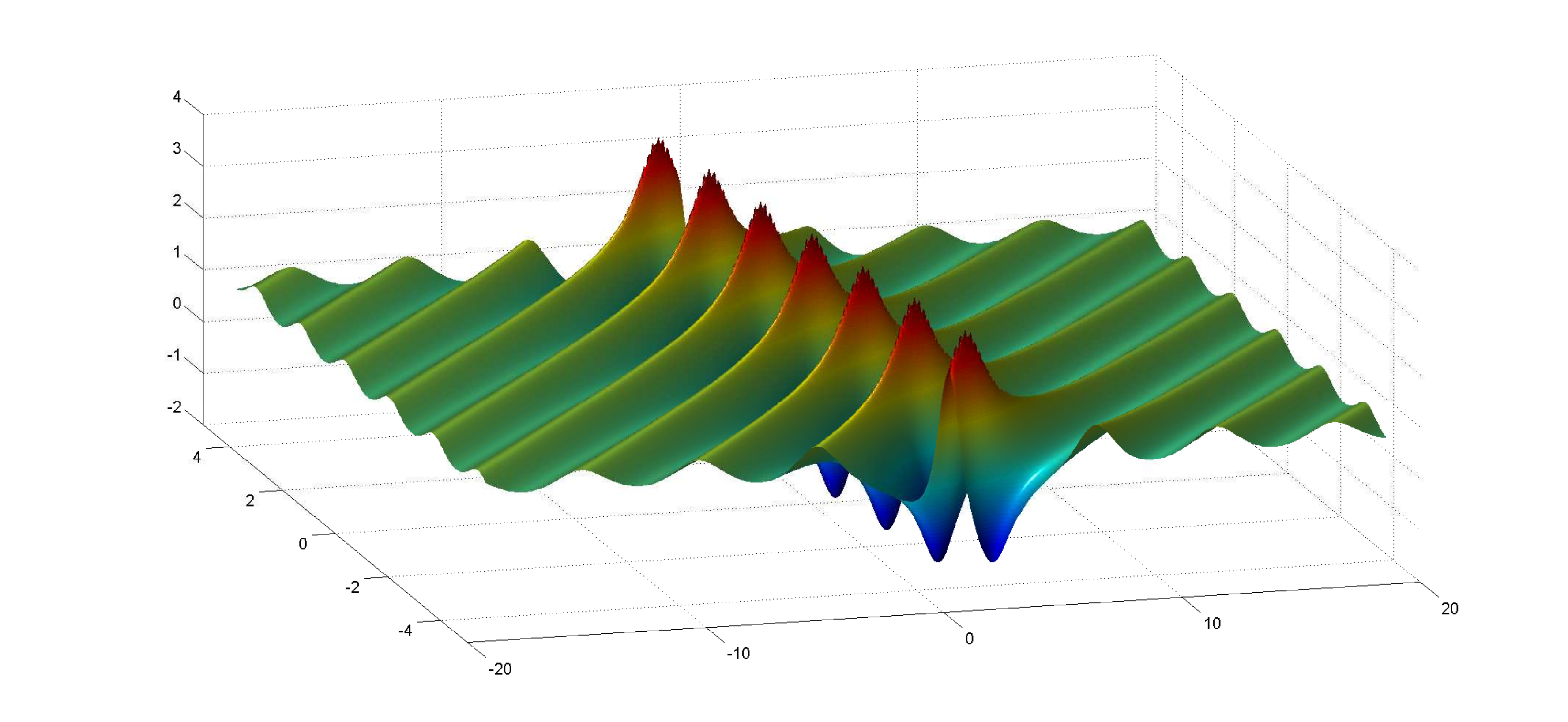}\\
\includegraphics[scale=0.4]{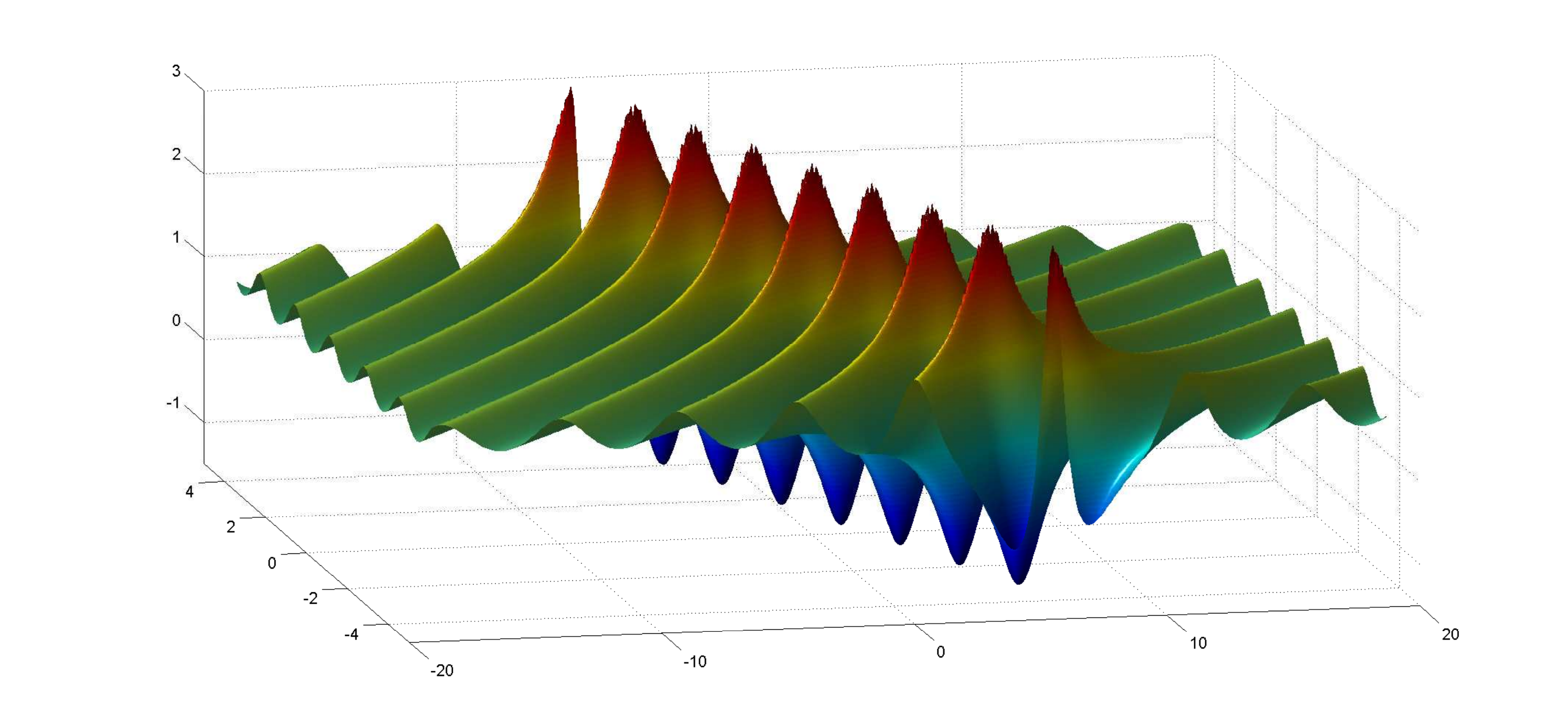}
\caption{Three outcomes of the one-fold transformation of the periodic solution (\ref{Jacob-1-intro}). } \label{fig-1-fold-1}
\end{figure}

Substituting (\ref{represent-new}) into the two-fold transformation (\ref{two-fold-general})
yields the new solution to the mKdV equation (\ref{mKdV}) in the form:
\begin{equation}
\label{two-fold-new}
\hat{u} = u + \frac{4 (\lambda_1^2-\lambda_2^2) (\lambda_1 N_1 D_2 - \lambda_2 N_2 D_1)}{(\lambda_1^2 + \lambda_2^2) D_1 D_2
- 8 \lambda_1 \lambda_2 N_1 N_2 - 2 \lambda_1 \lambda_2 S_1 S_2},
\end{equation}
where for $j = 1,2$
\begin{eqnarray*}
N_j & := & \frac{p_j q_j}{p_j^2 + q_j^2} \left[ (p_j^2 + q_j^2)^2 \phi_j^2 - 4 \right] + 2 (p_j^2 - q_j^2) \phi_j, \\
S_j & := & \frac{(p_j^2- q_j^2)}{p_j^2 + q_j^2} \left[ (p_j^2 + q_j^2)^2 \phi_j^2 - 4 \right] - 8 p_j q_j \phi_j, \\
D_j & := & (p_j^2 + q_j^2)^2 \phi_j^2 + 4.
\end{eqnarray*}
The expression for $\phi_1$ is given by (\ref{phi-final}) with $\psi_1 = 0$.
Note that the numerical factor $(\lambda_1^2 - \lambda_2^2)$ in (\ref{phi-final}) cancels with
the denominators in the expressions (\ref{relation-1-simple}), (\ref{relation-2-simple}),
and (\ref{relation-3-simple-better}) for the squared eigenfunctions. Therefore,
the expression for $\phi_2$ and $\phi_3$ can be obtained with the transformation
$\lambda_1 \mapsto \lambda_2$ and $\lambda_1 \mapsto \lambda_3$ respectively.

Three choices exist for an eigenvalue pair $(\lambda_1,\lambda_2)$ from the three real
eigenvalues in (\ref{eig-real-intro}). Figure \ref{fig-2-fold-1} shows three new solutions computed from
the periodic wave (\ref{Jacob-1-intro}) with the same choice of parameters as in (\ref{choice-1}).
The three solutions display three different choices of two algebraic solitons propagating on the
background of the periodic travelling wave.

\begin{figure}[ht]
\centering
\includegraphics[scale=0.4]{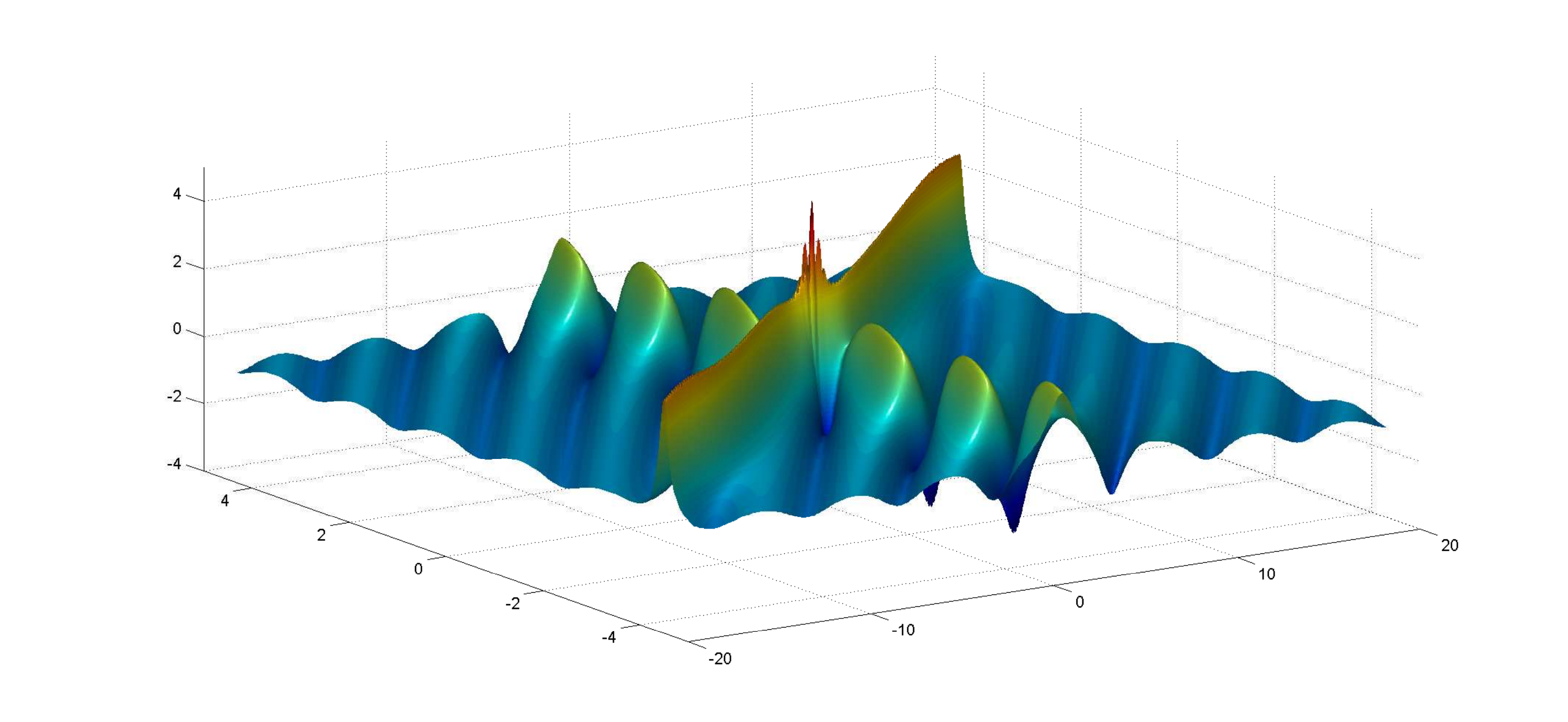}\\
\includegraphics[scale=0.4]{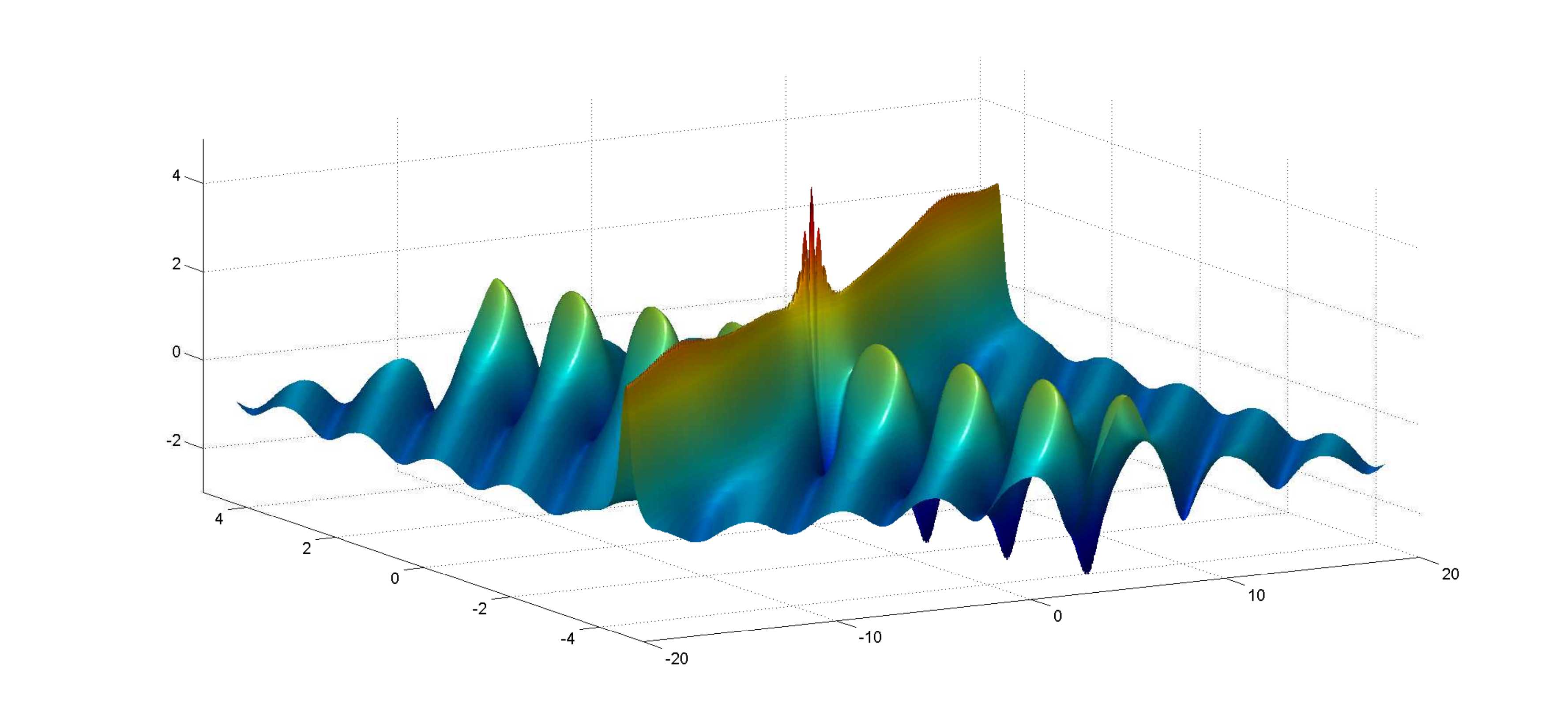}\\
\includegraphics[scale=0.4]{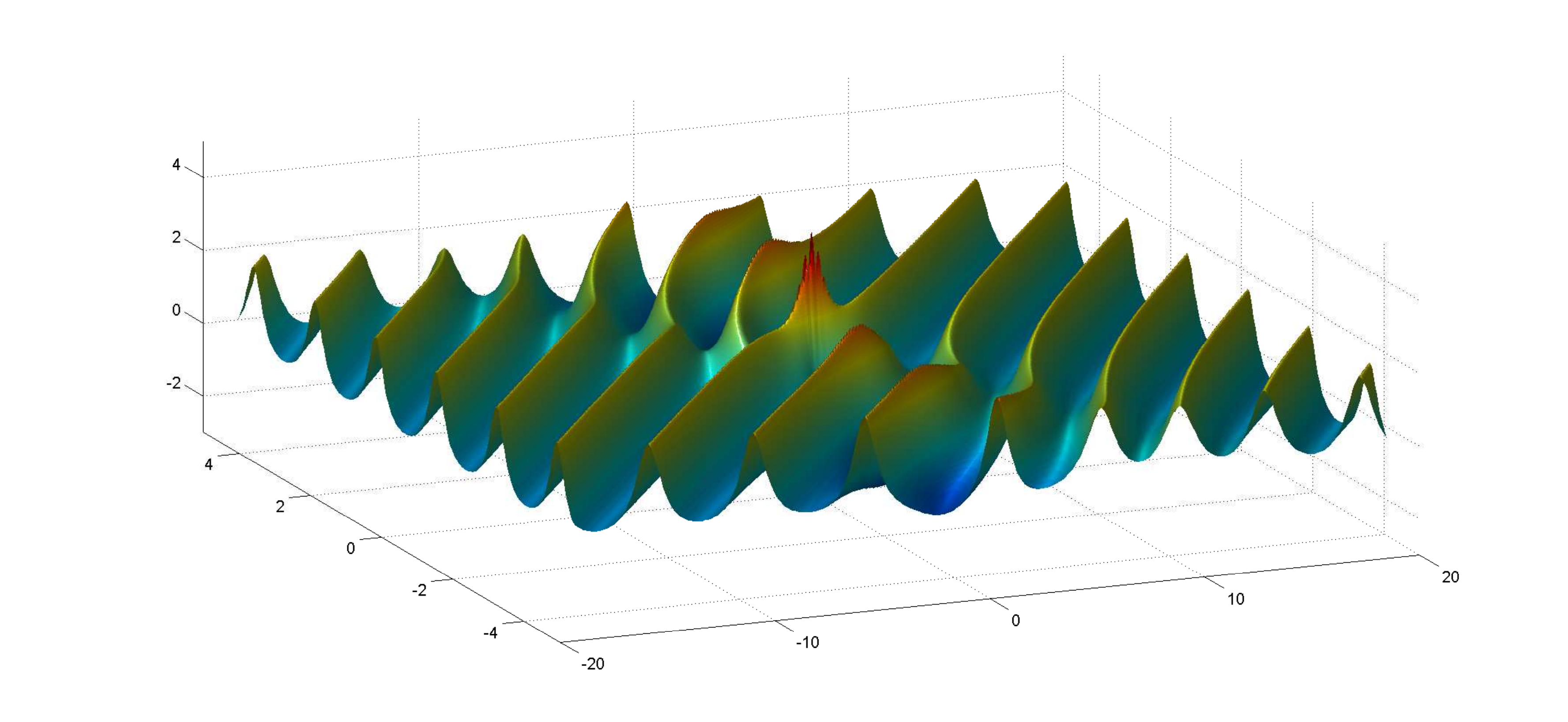}
\caption{Three outcomes of the two-fold transformation of the periodic solution (\ref{Jacob-1-intro}). } \label{fig-2-fold-1}
\end{figure}

Substituting (\ref{represent-new}) into the three-fold transformation (\ref{three-fold-general})
yields a new solution in the explicit form that is similar to (\ref{one-fold-new}) and (\ref{two-fold-new}).
There is only one choice for the three real eigenvalues $(\lambda_1,\lambda_2,\lambda_3)$ in (\ref{eig-real-intro}).
Figure \ref{fig-3-fold-1} shows the new solution computed from
the periodic wave (\ref{Jacob-1-intro}) with the same choice of parameters as in (\ref{choice-1}).
The new solution displays three algebraic solitons propagating on the background of the
periodic travelling wave.

\begin{figure}[ht]
\centering
\includegraphics[scale=0.5]{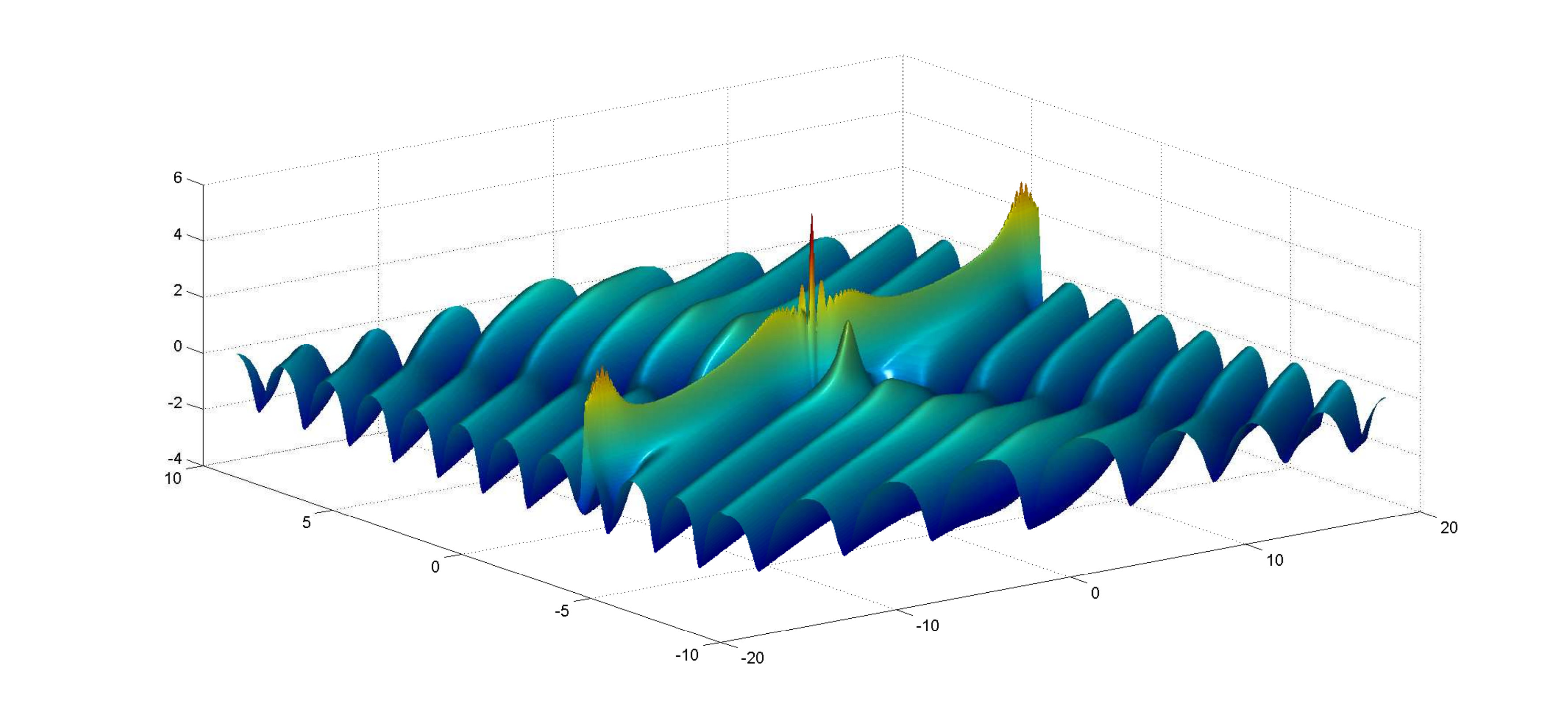}
\caption{Outcome of the three-fold transformation of the periodic solution (\ref{Jacob-1-intro}). } \label{fig-3-fold-1}
\end{figure}

The proof of Theorem \ref{theorem-rogue-dn} relies on the following two lemmas.

\begin{lemma}
\label{lemma-phi}
Let $(p_{1,2,3},q_{1,2,3})$ be the periodic solution to the Lax equations (\ref{3.2})--(\ref{3.3})
for $\lambda_{1,2,3}$ and $(\hat{p}_{1,2,3},\hat{q}_{1,2,3})$ be the second solution in Lemma \ref{lemma-time-second}.
Let $\tilde{u}$ be the Darboux transformations of $u$ with the periodic solution and
$\hat{u}$ be the corresponding Darboux transformations of $u$ with the second solution.
Then, we have
\begin{equation}
\label{limits-u}
\lim_{|\phi_{1,2,3}| \to \infty} \hat{u} = \tilde{u}, \quad
\lim_{|\phi_{1,2,3}| \to 0} \hat{u} = 2 u - \tilde{u}.
\end{equation}
\end{lemma}

\begin{proof}
By using (\ref{one-fold}) and (\ref{one-fold-new}) for the one-fold transformation,
we obtain
$$
\lim_{|\phi_1| \to \infty} \hat{u} = u + \frac{4 \lambda_1 p_1 q_1}{p_1^2 + q_1^2} = \tilde{u}
$$
and
$$
\lim_{|\phi_1| \to 0} \hat{u} = u - \frac{4 \lambda_1 p_1 q_1}{p_1^2 + q_1^2} = 2u - \tilde{u}.
$$
By using (\ref{two-fold}) and (\ref{two-fold-new}) for the two-fold transformation,
we obtain similarly:
$$
\lim_{|\phi_1|,|\phi_2| \to \infty} \hat{u} = u +
\frac{4 (\lambda_1^2-\lambda_2^2) \left[ \lambda_1 p_1 q_1 (p_2^2 + q_2^2) - \lambda_2 p_2 q_2 (p_1^2 + q_1^2)\right]}{
(\lambda_1^2 + \lambda_2^2) (p_1^2 + q_1^2)(p_2^2 + q_2^2) - 2 \lambda_1 \lambda_2 [ 4 p_1 q_1 p_2 q_2 + (p_1^2 - q_1^2)(p_2^2 - q_2^2)]}
= \tilde{u}
$$
and
$$
\lim_{|\phi_1|,|\phi_2| \to 0} \hat{u} = u -
\frac{4 (\lambda_1^2-\lambda_2^2) \left[ \lambda_1 p_1 q_1 (p_2^2 + q_2^2) - \lambda_2 p_2 q_2 (p_1^2 + q_1^2)\right]}{
(\lambda_1^2 + \lambda_2^2) (p_1^2 + q_1^2)(p_2^2 + q_2^2) - 2 \lambda_1 \lambda_2 [ 4 p_1 q_1 p_2 q_2 + (p_1^2 - q_1^2)(p_2^2 - q_2^2)]}
= 2u - \tilde{u}.
$$
The proof for the three-fold Darboux transformation (\ref{three-fold-general}) is similar.
\end{proof}

\begin{lemma}
\label{lemma-phi-growth-dn}
Assume $u_4 < u_3 < u_2 < u_1$ such that $u_1 + u_2 + u_3 + u_4 = 0$,
$u_1 + u_2 \neq 0$, $u_1 + u_3 \neq 0$, and $u_2 + u_3 \neq 0$.
Under the following three non-degeneracy conditions,
\begin{equation}
\label{nondeg1}
\oint \frac{\lambda_1^2 (u_1 u_2 + u_3 u_4 + 2 u^2)}{(e + 4 \lambda_1^2 u)^2} dy \neq 0,
\end{equation}
\begin{equation}
\label{nondeg2}
\oint \frac{\lambda_2^2 (u_1 u_3 + u_2 u_4 + 2 u^2)}{(e + 4 \lambda_2^2 u)^2} dy \neq 0,
\end{equation}
and
\begin{equation}
\label{nondeg3}
\oint \frac{\lambda_3^2 (u_1 u_4 + u_2 u_3 + 2 u^2)}{(e + 4 \lambda_3^2 u)^2} dy \neq 0,
\end{equation}
where $\oint$ denotes the mean value integral,
there exist $c_1, c_2, c_3 \neq c$ such that the second solutions to the Lax system (\ref{3.2})--(\ref{3.3})
for $u$ in (\ref{Jacob-1-intro}) and eigenvalues $\lambda_1,\lambda_2,\lambda_3$ in (\ref{eig-real-intro})
are linearly growing in $x$ and $t$ everywhere on the $(x,t)$ plane except for the straight lines
$x - c_{1,2,3} t = \xi_{1,2,3}$, where
$\xi_1,\xi_2,\xi_3$ are phase parameters which are not uniquely defined.
\end{lemma}

\begin{proof}
Thanks to the periodicity of $\varphi = (p_1,q_1)^t$ in $(x,t)$ and the explicit representation (\ref{represent-new})--(\ref{phi-final})
for the second solution $\varphi = (\hat{p}_1,\hat{q}_1)^t$, the proof follows from analysis of the factor
\begin{equation}
\label{tilde-phi}
\tilde{\phi}_1(x,t) = \int_0^{x-ct} \frac{\lambda_1^2 (c - 4 \lambda_1^2 - 2 u^2)}{(e + 4 \lambda_1^2 u)^2} dy + t + \psi_1.
\end{equation}
Let $M_1$ be the mean value of the factor under the integration sign and assume that $M_1 \neq 0$.
Then,
\begin{equation}
\label{tilde-phi-growth}
\tilde{\phi}_1(x,t) = M_1 (x-ct) + t + {\rm periodic \; function}.
\end{equation}
Hence, $|\tilde{\phi}_1(x,t)|\to \infty$ as a linear function of $(x,t)$ everywhere
on the $(x,t)$ plane except for the straight line $x - c_1 t = \xi_1$, where
$c_1 = c - M_1^{-1}$ and $\xi_1$ is not uniquely defined due to the phase parameter $\psi_1$.
Similar formulas are obtained for the second solutions $(\hat{p}_2,\hat{q}_2)$ and $(\hat{p}_3,\hat{q}_3)$.
The conditions $M_1, M_2, M_3 \neq 0$ are equivalent to the non-degeneracy conditions
(\ref{nondeg1}), (\ref{nondeg2}), and (\ref{nondeg3}) respectively, e.g.
$c - 4 \lambda_1^2 = 2 (\lambda_2^2 + \lambda_3^2 - \lambda_1^2) = -(u_1u_2 + u_3 u_4)$.
\end{proof}

\begin{remark}
The first limit in (\ref{limits-u}) tells us that the new solution $\hat{u}$ approaches
the transformed periodic wave $\tilde{u}$ almost everywhere on the $(x,t)$ plane as $|x| + |t| \to \infty$
except for the lines $x-c_{1,2,3}t = \xi_{1,2,3}$. Algebraic solitons propagate along these lines with the wave speeds $c_{1,2,3}$.
\end{remark}

\begin{remark}
The second limit in (\ref{limits-u}) tells us the amplitude of the algebraic soliton on the
periodic background $\tilde{u}$. The periodic wave $u$ in (\ref{Jacob-1}) has two extremal
points $u(0) = u_1$ and $u(L/2) = u_2$ and our convention is $u_4 < u_3 < u_2 < u_1$.
Table \ref{Table1} shows the periodic background $\tilde{u}$
and the new solution $\hat{u}$ at two extremal points of $u$.
\end{remark}

\begin{table}[ht]
\begin{center}
\begin{tabular}{|c|c|c|c|c|}
\hline
Transformation & $\tilde{u}(0)$ & $\hat{u}(0)$ & $\tilde{u}(L/2)$ & $\hat{u}(L/2)$  \\
\hline
one-fold with $\lambda_1$ & $-u_2$ & $2u_1 + u_2$ & $-u_1$ & $2 u_2 + u_1$\\
one-fold with $\lambda_2$ & $-u_3$ & $2u_1 + u_3$ & $-u_4$ & $2 u_2 + u_4$ \\
one-fold with $\lambda_3$ & $-u_4$ & $2u_1 + u_4$ & $-u_3$ & $2 u_2 + u_3$\\
two-fold with $(\lambda_1,\lambda_2)$ & $u_4$ & $2u_1 - u_4$ & $u_3$ & $2 u_2 - u_3$\\
two-fold with $(\lambda_1,\lambda_3)$ & $u_3$ & $2u_1 - u_3$ & $u_4$ & $2 u_2 - u_4$\\
two-fold with $(\lambda_2,\lambda_3)$ & $u_2$ & $2u_1 - u_2$ & $u_1$ & $2 u_2 - u_1$\\
three-fold with $(\lambda_1,\lambda_2,\lambda_3)$ & $-u_1$ & $3u_1$ & $-u_2$ & $3 u_2$\\
\hline
\end{tabular}
\end{center}
\caption{Characteristics of rogue waves on the periodic background (\ref{Jacob-1-intro}).}
\label{Table1}
\end{table}

We next represent graphically the outcomes of the one-fold, two-fold, and three-fold Darboux transformations
for the periodic wave (\ref{Jacob-2-intro}). For the one-fold transformation (\ref{one-fold-new}), only one choice
generates the real solution $\hat{u}$ to the mKdV equation (\ref{mKdV}). This choice
corresponds to the real eigenvalue $\lambda_3$ in (\ref{eig-complex-intro}).
Figure \ref{fig-1-fold-2} shows the new solution computed from
the periodic wave (\ref{Jacob-2-intro}) with the parameter values:
\begin{equation}
\label{choice-2}
a = 1.5, \quad b = -0.5, \quad \alpha = -0.5, \quad \beta = 2.
\end{equation}
The new solution displays the algebraic soliton propagating on the background of
the periodic travelling wave.

\begin{figure}[ht]
\centering
\includegraphics[scale=0.5]{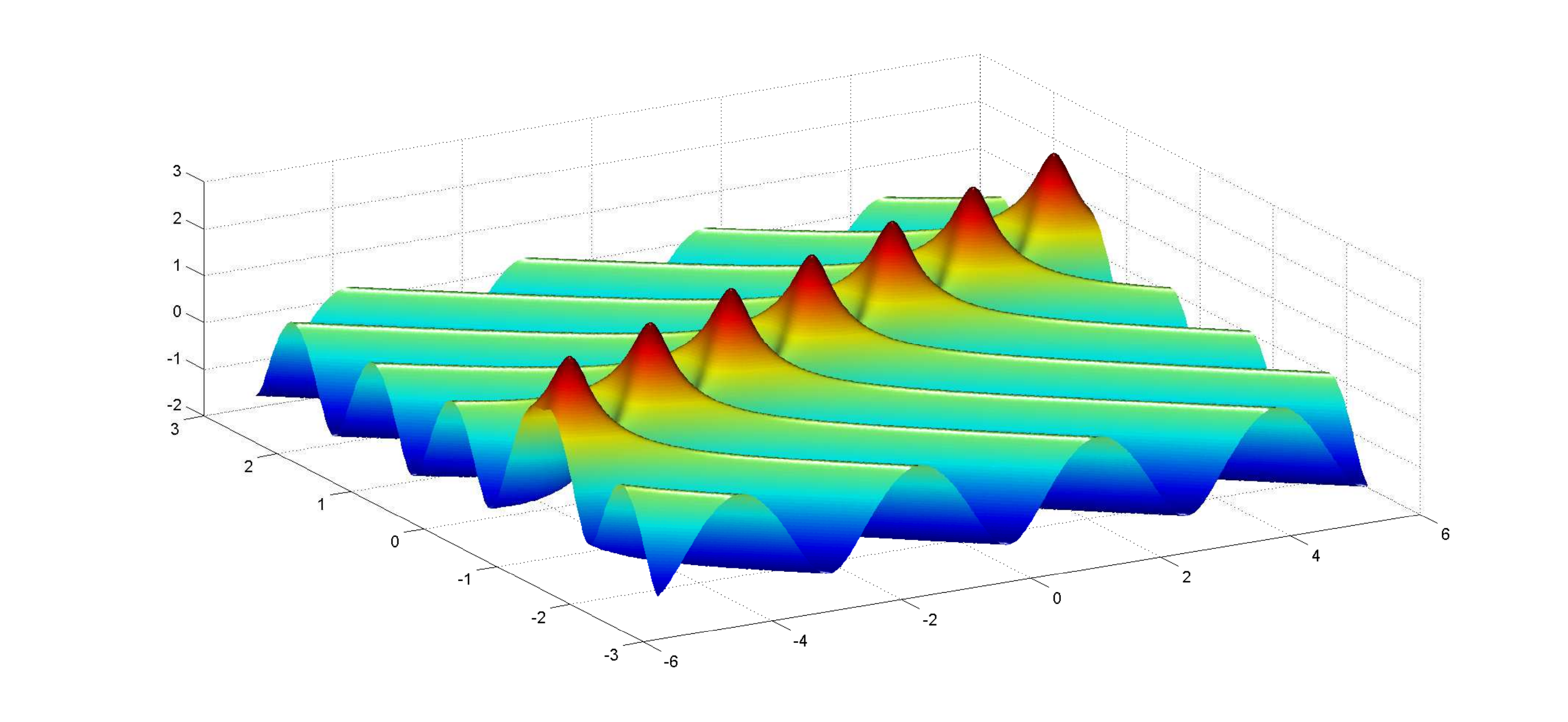}
\caption{Outcome of the one-fold transformation of the periodic solution (\ref{Jacob-2-intro}). } \label{fig-1-fold-2}
\end{figure}

For the two-fold transformation (\ref{two-fold-new}), only one choice
generates the real solution $\hat{u}$ to the mKdV equation (\ref{mKdV}). This choice
corresponds to the complex-conjugate eigenvalues $(\lambda_1,\lambda_2)$ in (\ref{eig-complex-intro}).
Figure \ref{fig-2-fold-2} shows the new solution computed from
the periodic wave (\ref{Jacob-2-intro}) with the same choice of parameters as in (\ref{choice-2}).
The new solution displays a fully localized rogue wave on the background of the periodic travelling wave.

\begin{figure}[ht]
\centering
\includegraphics[scale=0.5]{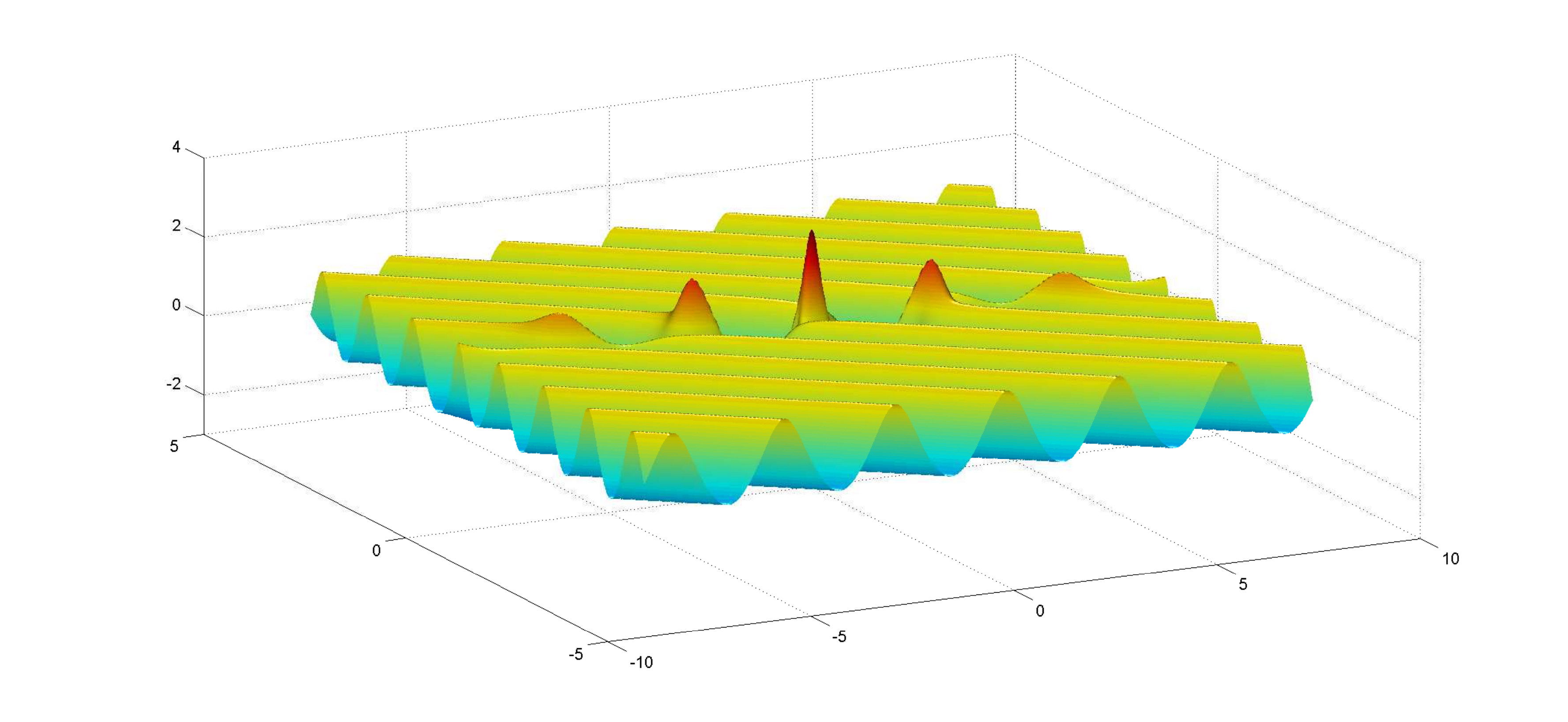}
\caption{Outcome of the two-fold transformation of the periodic solution (\ref{Jacob-2-intro}). } \label{fig-2-fold-2}
\end{figure}

Figure \ref{fig-3-fold-2} shows the outcome of the three-fold transformation
computed from the periodic wave (\ref{Jacob-2-intro}) with the same choice of parameters as in (\ref{choice-2}).
All three eigenvalues $(\lambda_1,\lambda_2,\lambda_3)$ are selected from the list (\ref{eig-complex-intro}).
The new solution displays both the algebraic soliton propagating on the background of the
periodic travelling wave and a localized rogue wave at the center.

\begin{figure}[ht]
\centering
\includegraphics[scale=0.5]{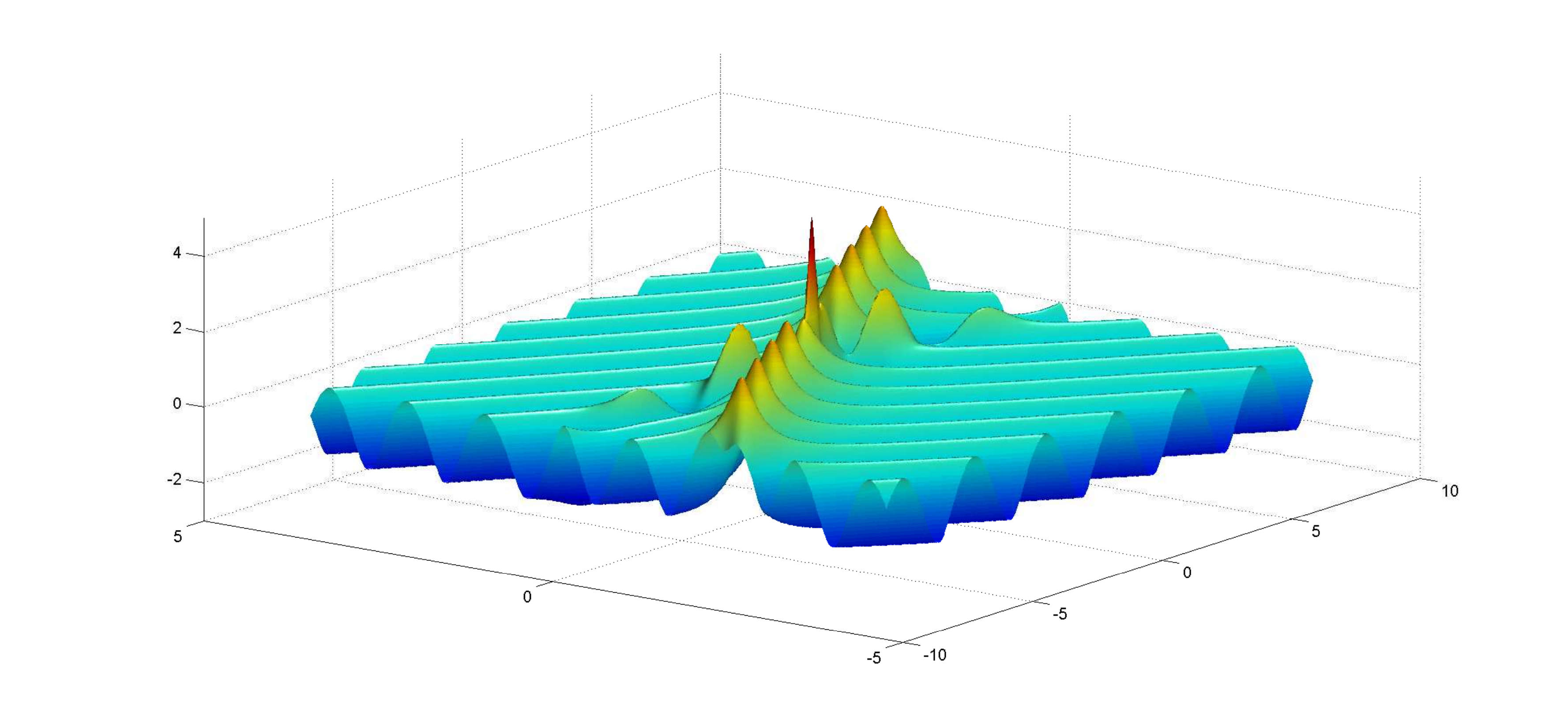}
\caption{Three-fold transformation of the periodic solution (\ref{Jacob-2}). } \label{fig-3-fold-2}
\end{figure}

The proof of Theorem \ref{theorem-rogue-cn} relies on Lemma \ref{lemma-phi} and the following
lemma.

\begin{lemma}
\label{lemma-phi-growth-cn}
Assume $a \neq \pm b$. Under the following two non-degeneracy conditions,
\begin{equation}
\label{nondeg4}
{\rm Im} \oint \frac{\lambda_1^2 (2 i \beta (a-b) - (a+b)^2 + 4 u^2)}{(e + 4 \lambda_1^2 u)^2} dy \neq 0
\end{equation}
and
\begin{equation}
\label{nondeg5}
\oint \frac{\lambda_3^2 (4 \beta^2 + a^2 + b^2 + 6ab + 8 u^2)}{(e + 4 \lambda_3^2 u)^2} dy \neq 0,
\end{equation}
the second solutions to the Lax system (\ref{3.2})--(\ref{3.3}) for
$u$ in (\ref{Jacob-2-intro}) and eigenvalues $\lambda_1,\lambda_2$ in (\ref{eig-complex-intro}) are linearly
growing in $x$ and $t$ everywhere, whereas there exists $c_0 \neq c$ such that the second solution
for $\lambda_3$ in (\ref{eig-complex-intro}) is linearly growing in $x$ and $t$ everywhere
except for the straight line $x - c_0 t = \xi_0$, where $\xi_0$ is the phase parameter
which is not uniquely defined.
\end{lemma}

\begin{proof}
In the representations (\ref{tilde-phi}) and (\ref{tilde-phi-growth}), the value of $M_1$ is complex
since $\lambda_1$ in (\ref{eig-complex-intro}) is complex. If ${\rm Im}(M_1) \neq 0$, then
$|\tilde{\phi}_1(x,t)|\to \infty$ as a linear function of $(x,t)$ everywhere
on the $(x,t)$ plane. The same is true for the solution $(\hat{p}_2,\hat{q}_2)$
since $\lambda_2 = \bar{\lambda}_1$. However, the solution $(\hat{p}_3,\hat{q}_3)$
with real $\lambda_3$ is real and the representation (\ref{tilde-phi-growth}) with $M_3 \neq 0$
shows that $|\tilde{\phi}_1(x,t)|$ remains bounded on the straight line $x - c_0 t = \xi_0$, where
$c_0 = c - M_3^{-1}$ and $\xi_0$ is not uniquely defined due to the phase parameter $\psi_3$.
The conditions ${\rm Im}(M_1) \neq 0$ and $M_3 \neq 0$ are equivalent to the non-degeneracy conditions
(\ref{nondeg4}) and (\ref{nondeg5}) respectively, e.g.
$c - 4 \lambda_3^2 = 2 (\lambda_1^2 + \lambda_2^2 - \lambda_3^2) = -(\beta^2 + (a^2 + b^2 + 6ab)/4)$.
\end{proof}

\begin{remark}
By Lemma \ref{lemma-phi-growth-cn}, the new solution $\hat{u}$ after
the one-fold transformation with eigenvalue $\lambda_3$ approaches
the transformed periodic wave $\tilde{u}$ everywhere on the $(x,t)$ plane as $|x| + |t| \to \infty$
except for the line $x-c_0 t = \xi_0$, where the algebraic soliton propagates.
The periodic wave $u$ in (\ref{Jacob-2}) has two extremal
points $u(0) = a$ and $u(L/2) = b$ and our convention is $b < a$.
Table \ref{Table2} shows the periodic background $\tilde{u}$
and the new solution $\hat{u}$ at two extremal points of $u$.
\end{remark}

\begin{remark}
By Lemma \ref{lemma-phi-growth-cn}, the new solution $\hat{u}$ after
the two-fold transformation with eigenvalues $(\lambda_1,\lambda_2)$ approaches
the transformed periodic wave $\tilde{u}$ everywhere on the $(x,t)$ plane as $|x| + |t| \to \infty$.
Hence, this is the proper rogue wave in the sense of Definition \ref{def-rogue}.
Table \ref{Table2} shows the magnification factors of the rogue wave
defined as a ratio between the extremal values of $\hat{u}$ and $\tilde{u}$, e.g.
$M = \max\{|2a-b|/|b|,|2b-a|/|a|\}$.
\end{remark}

\begin{remark}
The three-fold transformation with all three eigenvalues $(\lambda_1,\lambda_2,\lambda_3)$
produces both the rogue wave and the algebraic soliton with the wave speed $c_0$.
Table \ref{Table2} shows that the new wave $\hat{u}$ has triple magnification compared to the
periodic background $\tilde{u}$.
\end{remark}

\begin{table}[ht]
\begin{center}
\begin{tabular}{|c|c|c|c|c|}
\hline
Transformation & $\tilde{u}(0)$ & $\hat{u}(0)$ & $\tilde{u}(L/2)$ & $\hat{u}(L/2)$ \\
\hline
one-fold with $\lambda_3$ & $-b$ & $2a + b$ & $-a$ & $2 b + a$\\
two-fold with $(\lambda_1,\lambda_2)$ & $b$ & $2a - b$ & $a$ & $2 b - a$\\
three-fold with $(\lambda_1,\lambda_2,\lambda_3)$ & $-a$ & $3a$ & $-b$ & $3 b$\\
\hline
\end{tabular}
\end{center}
\caption{Characteristics of rogue waves on the periodic background (\ref{Jacob-2-intro}).}
\label{Table2}
\end{table}

\section{Conclusion}
\label{sec-8}

We have addressed the most general periodic travelling wave solutions to the mKdV equation
(\ref{mKdV}) and obtained the periodic solutions to the Lax system (\ref{3.2})--(\ref{3.3})
for three particular pairs of eigenvalues $\lambda$ away from the imaginary axis.
For the family of periodic waves (\ref{Jacob-1-intro}) generalizing the
{\em dn}-periodic wave (\ref{dn-intro}), the three pairs of eigenvalues
$\pm \lambda_1$, $\pm \lambda_2$, $\pm \lambda_3$ are real.
For the family of periodic waves (\ref{Jacob-2-intro}) generalizing the {\em cn}-periodic
wave (\ref{cn-intro}), one pair $\pm \lambda_3$ is real and two pairs $\pm \lambda_1$, $\pm \lambda_2$
form a quadruplet of complex eigenvalues.
By using the Darboux transformations, we have showed that transformations
involving the periodic eigenfunctions remain in the class of the same periodic wave solutions,
whereas transformations involving second nonperiodic solutions to the Lax system
(\ref{3.2})--(\ref{3.3}) for the same eigenvalues generate new solutions
on the background of the periodic waves. Among new solutions, one solution is a rogue wave
on the periodic background satisfying (\ref{rogue-wave-def}),
whereas all others are algebraic solitons propagating on the periodic background.
The rogue wave exist on the background of the periodic wave (\ref{Jacob-2-intro})
which is expected to be modulationally unstable with respect to perturbations of long periods.

Let us summarize the outcomes of the algebraic method
for periodic solutions to the mKdV equation (\ref{mKdV})
expressed by the Riemann Theta functions of genus $g$ with $g = 0,1,2$.
These solutions are obtained by degeneration of the three Dubrovin variables
$\mu_1$, $\mu_2$, and $\mu_3$ related to the periodic solution $u$
in (\ref{mu1}), (\ref{mu2}), and (\ref{mu3}).

\begin{itemize}
\item[$g=0$] If $\mu_1 = \mu_2 = \mu_3 = 0$, then $u(x,t) = u_1$ is a constant wave with $u_1 \in \mathbb{R}$.
This corresponds to only one pairs of real eigenvalues $\pm \lambda_1$ with $\lambda_1 = u_1$ and
the constant solution to the Lax system (\ref{3.2})--(\ref{3.3}).

\item[$g=1$] If $\mu_2 = \mu_3 = 0$, then $u(x,t) = u(x-ct)$ satisfies the second-order equation
$$
\frac{d^2 u}{dx^2} + 2 u^3 - c u = 0,
$$
which is solved by two families of the solutions (\ref{Jacob-1-intro}) and (\ref{Jacob-2-intro})
in the case $e = 0$. One solution corresponds to $u_4 = -u_1$, $u_3 = -u_2$ and it generalizes the
{\em dn}-periodic wave (\ref{dn-intro}), whereas the other solution corresponds to $b = -a$, $\alpha = 0$, $\beta \neq 0$
and it generalizes the {\em cn}-periodic wave (\ref{cn-intro}). As is shown in \cite{CPkdv}, the algebraic method
produces only two pairs of eigenvalues $\pm \lambda_1$, $\pm \lambda_2$
with the periodic eigenfunctions of the Lax equations (\ref{3.2})--(\ref{3.3}),
where
$$
\lambda_1 = \frac{1}{2}(u_1+u_2), \quad \lambda_2 = \frac{1}{2}(u_1-u_2)
$$
for the periodic solution (\ref{Jacob-1-intro}) and
$$
\lambda_1 = \frac{1}{2}(a + i \beta), \quad \lambda_2 = \frac{1}{2}(a - i \beta)
$$
for the periodic solution (\ref{Jacob-2-intro}).

\item[$g=2$] If $\mu_3 = 0$, then $u(x,t) = u(x-ct)$ satisfies the third-order equation
$$
\frac{d^3 u}{dx^3} + 6 u^2 \frac{du}{dx} - c \frac{du}{dx} = 0,
$$
which is solved by two families of the solutions (\ref{Jacob-1-intro}) and (\ref{Jacob-2-intro})
in the general case $e \neq 0$.
As is shown here, the algebraic method produces only three pairs of eigenvalues $\pm \lambda_1$, $\pm \lambda_2$, $\pm \lambda_3$
with the periodic eigenfunctions of the Lax equations (\ref{3.2})--(\ref{3.3}).
\end{itemize}

Based on the summary above, it is natural to conjecture that the solution $u$ to the mKdV equation
(\ref{mKdV}) expressed by quasi-periodic Riemann Theta function of genus $g$
is related to exactly $g+1$ pairs of eigenvalues in the spectral problem (\ref{3.2}) with
the quasi-periodic eigenfunctions of the same periods. Moreover, location of these eigenvalues is related
to parameters of the Riemann Theta functions. It is also natural
to conjecture that no other eigenvalues $\lambda$ with the quasi-periodic eigenfunctions exist away from the imaginary axis.
To the best of our knowledge, these mathematical questions have not been solved in the literature, in spite of the large amount of
publications on the mKdV equation. Solving these problems in future looks an interesting
question of fundamental significance with many potential applications.

\vspace{0.5cm}

{\bf Acknowledgements.} The work of J.C. was supported by the National Natural Science Foundation
of China (No.11471072). The work of D.P. is supported by the State task program in the sphere
of scientific activity of Ministry of Education and Science of the Russian Federation
(Task No. 5.5176.2017/8.9) and from the grant of President of Russian Federation
for the leading scientific schools (NSH-2685.2018.5).

\end{document}